\documentclass[conference]{cls/IEEEtran}
\let\labelindent\relax
\usepackage{times}
\usepackage[numbers]{natbib}
\usepackage[table,usenames,dvipsnames]{xcolor}      
\usepackage{extarrows}                              
\usepackage{enumitem}                               
\usepackage{amsmath,amssymb,amsfonts,amsthm,dsfont} 
\usepackage{algorithm,algorithmicx,listings}        
\usepackage[noend]{algpseudocode}			        

\usepackage{mathtools} 

\usepackage{graphicx,tabularx}
\usepackage[export]{adjustbox}
\usepackage[font={small}]{caption}
\usepackage[breaklinks=true, colorlinks, bookmarks=true, citecolor=Black, urlcolor=Violet,linkcolor=Black]{hyperref}

\newtheorem{theorem}{Theorem}
\newtheorem{proposition}{Proposition}

\newtheorem{lemma}{Lemma}
\theoremstyle{definition}
\newtheorem{definition}{Definition}
\newtheorem{assumption}{Assumption}
\newtheorem*{assumption*}{Assumption}
\newtheorem*{problem*}{Problem}
\newtheorem{problem}{Problem}
\theoremstyle{remark}

\newtheorem*{solution*}{Solution}

\newcommand{\calD}{{\cal D}}

\newcommand{\calF}{{\cal F}}
\newcommand{\calG}{{\cal G}}

\newcommand{\calO}{{\cal O}}

\newcommand{\calS}{{\cal S}}

\newcommand{\calX}{{\cal X}}

\newcommand{\bfu}{\mathbf{u}}

\newcommand{\bfx}{\mathbf{x}}
\newcommand{\bfy}{\mathbf{y}}

\newcommand{\bfrho}{\boldsymbol{\rho}}

\newcommand{\bbR}{\mathbb{R}}
\newcommand{\bbS}{\mathbb{S}}

\newcommand{\scaleLine}[2][1]{\resizebox{#1\linewidth}{!}{#2}}
\newcommand{\scaleMathLine}[2][1]{\resizebox{#1\linewidth}{!}{$\displaystyle{#2}$}}
\newcommand{\prl}[1]{\left(#1\right)}
\newcommand{\brl}[1]{\left[#1\right]}
\newcommand{\crl}[1]{\left\{#1\right\}}

\DeclarePairedDelimiterX{\norm}[1]{\lVert}{\rVert}{#1}

\newcommand*{\intset}[1]{[#1]_\mathbb{N}}
\newcommand*{\col}[1]{\textbf{col}\left(#1\right)}

\newcommand*{\Bb}{\mathbf{b}}
\newcommand*{\Bc}{\mathbf{c}}

\newcommand*{\Bf}{\mathbf{f}}
\newcommand*{\Bg}{\mathbf{g}}
\newcommand*{\Bh}{\mathbf{h}}

\newcommand*{\Bm}{\mathbf{m}}
\newcommand*{\Bn}{\mathbf{n}}

\newcommand*{\Bp}{\mathbf{p}}
\newcommand*{\Bq}{\mathbf{q}}
\newcommand*{\Br}{\mathbf{r}}
\newcommand*{\Bs}{\mathbf{s}}

\newcommand*{\Bu}{\mathbf{u}}
\newcommand*{\Bv}{\mathbf{v}}
\newcommand*{\Bw}{\mathbf{w}}
\newcommand*{\Bx}{\mathbf{x}}
\newcommand*{\By}{\mathbf{y}}
\newcommand*{\Bz}{\mathbf{z}}

\newcommand*{\BA}{\mathbf{A}}
\newcommand*{\BB}{\mathbf{B}}
\newcommand*{\BC}{\mathbf{C}}

\newcommand*{\BG}{\mathbf{G}}

\newcommand*{\BI}{\mathbf{I}}

\newcommand*{\BK}{\mathbf{K}}

\newcommand*{\BM}{\mathbf{M}}

\newcommand*{\BP}{\mathbf{P}}
\newcommand*{\BQ}{\mathbf{Q}}

\newcommand*{\BS}{\mathbf{S}}
\newcommand*{\BT}{\mathbf{T}}

\newcommand*{\Bxi}{\boldsymbol{\xi}}

\newcommand*{\Bzero}{\mathbf{0}}

\newcommand*{\C}[1]{\mathcal{#1}}

\newcommand*{\CA}{\mathcal{A}}

\newcommand*{\CE}{\mathcal{E}}
\newcommand*{\CF}{\mathcal{F}}
\newcommand*{\CL}{\mathcal{L}}

\newcommand*{\CO}{\mathcal{O}}

\newcommand*{\CX}{\mathcal{X}}

\newcommand*{\LS}{\mathcal{LS}}

\newcommand*{\gp}{\Bg}

\newcommand*{\dgp}{\dot{\Bg}}
\newcommand*{\lpg}{\bar{\Bg}}

\newcommand*{\DeltaE}{\Delta E}

\newcommand*{\carlengthtt}{l}

\newcommand*{\carspeed}{v}
\newcommand*{\carsteering}{\delta}
\newcommand*{\carorientation}{\psi}

\newcommand*{\dcarspeed}{\dot{v}}
\newcommand*{\dcarsteering}{\dot{\delta}}
\newcommand*{\dcarorientation}{\dot{\carorientation}}

\newcommand*{\stateAck}{\Bx}

\newcommand*{\stateRgs}{\Bs}

\newcommand*{\controlAck}{\Bu}

\newcommand*{\outputAck}{\By}

\newcommand*{\dddoutputAck}{\dddot{\By}}

\newcommand*{\nPath}{m}
\newcommand*{\nRbtState}{n}
\newcommand*{\nInput}{m}
\newcommand*{\nOutput}{m}

\newcommand*{\tRbt}{\tilde{\Bz}} 
\newcommand*{\dtRbt}{\dot{\tRbt}} 

\begin{document}
\title{Safe Robot Navigation in Cluttered Environments using Invariant Ellipsoids and a Reference Governor}
\author{Zhichao Li\quad Thai Duong\quad  Nikolay Atanasov
	\thanks{The authors are with the Department of Electrical and Computer Engineering, University of California, San Diego, La Jolla, CA 92093, USA {\tt\small \{zhl355,natanasov\}@eng.ucsd.edu}.}%
}


\maketitle

\begin{abstract}
This paper considers the problem of safe autonomous navigation in unknown environments, relying on local obstacle sensing. We consider a control-affine nonlinear robot system subject to bounded input noise and rely on feedback linearization to determine ellipsoid output bounds on the closed-loop robot trajectory under stabilizing control. A virtual governor system is developed to adaptively track a desired navigation path, while relying on the robot trajectory bounds to slow down if safety is endangered and speed up otherwise. The main contribution is the derivation of theoretical guarantees for safe nonlinear system path-following control and its application to autonomous robot navigation in unknown environments. 
\end{abstract}

\IEEEpeerreviewmaketitle

\section{Introduction}
\label{sec:introduction}
Safe autonomous operation in unstructured, dynamic, and a priori unknown environments is a cornerstone problem of robotics. Achieving safe autonomous navigation with ground and aerial robots will have transformative impact on transportation, structure inspection, and environmental monitoring. Similarly, achieving safe autonomous manipulation may have transformative impact on product assembly, construction, and medical applications. Research in motion planning and dynamical system control has led to significant progress in these directions. Geometric motion planning algorithms, including graph search~\cite{ARAstar,MHA} and sampling-based techniques~\cite{RRT,RRTstar,BITstar}, generate safe shortest paths in complex high-dimensional problems. However, ensuring that the planned paths are dynamically feasible and remain safe when a real-time tracking controller is employed on the robot system is an active area of research. Instead of geometric paths, kinodynamic planning~\cite{kinodynamic_rrts_webb2013,SST} aims to find high-order dynamically feasible trajectories directly simplifying the real-time control task. However, is challenging to plan trajectories over high-order variables such as angular acceleration or linear jerk, especially for robots with many degrees of freedom. Moreover, deciding the right trajectory clearance or time allocation during planning are major challenges in dynamically changing environments and in the presence of disturbances. These considerations make current techniques computationally challenging and, yet, safety guarantees when the planner is coupled with closed-loop controller are difficult to obtain.

The focus of this paper is a control design for nonlinear control-affine robot systems that provides theoretical guarantees for safe trajectory tracking when input disturbances are present, the environment is unknown and perceived only through local onboard sensing, and a only low-order geometric path is provided. Our approach relies on feedback linearization to determine invariant ellipsoid bounds on the closed-loop robot trajectory under stabilizing control. In contrast to existing feedback-linearization procedures, we examine the impact of disturbances carefully and develop semi-definite programming (SDP) and Lyapunov-equation techniques to tightly and computationally efficiently bound the worst-case system behavior. Inspired by \emph{reference governor} control techniques~\cite{garone2016_ERG,kolmanovsky2014ref_cmd_gov,Gov_ICRA17}, we develop a first-order virtual system to guide the real robot along a desired geometric path. The governor serves as a local equilibrium point for the robot that is able to move if safety is not endangered. In detail, the governor evaluates risk according to the volume of the local free space and the robot activeness, measured by the volume of the ellipsoid trajectory bounds, and adaptively selects its path-tracking speed to guarantee that the real robot remains safe and stable.

\begin{figure}[t]
\includegraphics[width=1\linewidth]{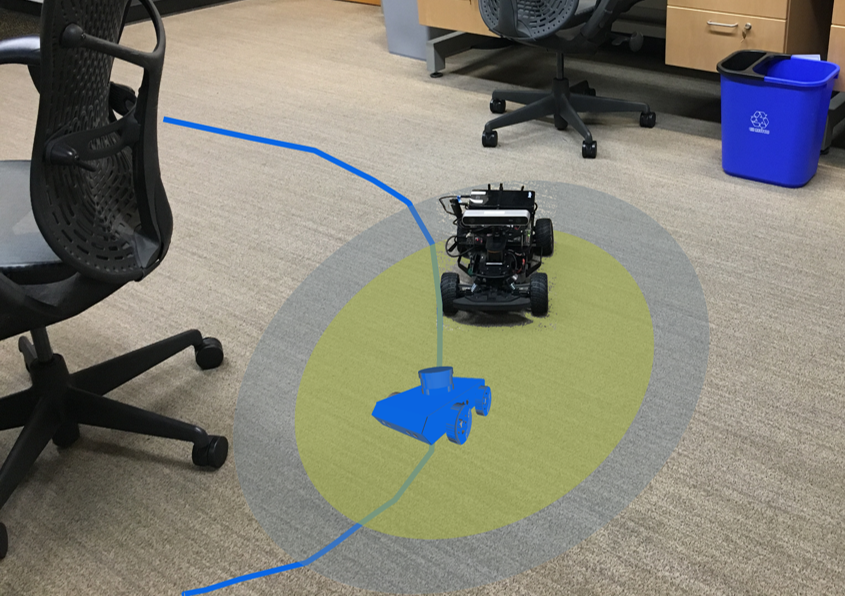}
\caption{An Ackermann-drive is regulated via feedback linearization towards a virtual first-order governor system (blue robot). An invariant ellipsoid (yellow) that bounds the possible robot trajectories defines a local safe zone. The governor system tracks a desired geometric path (blue curve) adaptively, slowing down if the local safe zone is approaching the limit of the local free space (grey ellipse) and speeding up otherwise.}
\label{fig:overview}
\end{figure}

\subsection{Related Work}
Safe corridor construction~\cite{SFC,SFC_FM} during geometric motion planning is a promising approach for generating high-order dynamically feasible trajectories. Existing approaches build the largest possible safe region around a geometric path using a series of connected convex sets (e.g., ellipsoids or polyhedra) and quadratic programming is used to generate safe dynamically feasible trajectories. A sampling-based Learning Model Predictive Controller is proposed in~\cite{rosolia2019sample} to build a safe set by iteratively refining a convex hull approximation of the reachable set of a linear system subject to bounded disturbance. The reachable set is approximated by finite-horizon closed-loop trajectories with a sampled disturbance sequence. A Q-value function, which behaves like a robust control Lyapunov function, and an associated control policy are obtained from the sampled trajectories.

To guarantee safety without sampling, most existing works rely on Lyapunov theory and reachable set over-approximations. Building on the seminal work of~\cite{burridge1999sequential}, sequential composition of funnels~\cite{funnel_idea} (reachable set approximations) offers effective means of guaranteeing safe navigation~\cite{lqr_tree_tedrake2009, Funnel_lib, gawron2017vfo, gawron_2018IROS_VFO}. Using sum-of-squares (SOS) optimization~\cite{sos_stability}, these techniques can deal with nonlinear systems directly and may handle nonholonomic constraints and bounded disturbances. While funnels are accurate in representing the system dynamics, the design of a funnel library faces a trade-off between the size of the library and computational complexity of composing the funnels online. A trajectory independent tracking controller was proposed by~\cite{singh2017robust} based on contraction theory. In the paper, a robust control invariant tube is constructed by minimizing a certain control contraction metric using SOS programming. Using Hamilton-Jacobi reachability, one can deal with nonlinear systems that are not control-affine. Herbert et al.~\cite{herbert2017fastrack} proposed a hybrid controller, which consists of a regular trajectory tracking controller and a safety controller based on reachability analysis. The reachability analysis is very accurate leading to fast and safe tracking performance but too computationally demanding to be performed online, requiring an a priori known environment. Control barrier function methods~\cite{ADA-SC-ME-GM-KS-PT:19,CBF_ames2014rapidly, CBF_wu2015safety, CBF_ames2016control, CBF_quadrotor} have gained significant attention because they allow including safety constraints into a quadratic program optimization of the control input without requiring feedback-linearization or reachability set approximations.

The reference governor framework~\cite{garone2016_ERG, kolmanovsky2014ref_cmd_gov,Gov_ICRA17} is a different approach to enforcing safety constraints that modulates the motion of a virtual system while the real pre-stabilized system tracks its motion. Arslan and Koditschek~\cite{Gov_ICRA17} developed a reference-governor controller for a double-integrator system navigating among spherical obstacles. The key idea is to construct a Lyapunov function in the sum-of-squares form and ensure safety by regulating the motion of the governor based on the size of the Lyapunov invariant set. Our work extends this idea significantly by applying it to general feedback-linearizable control-affine systems~\cite{fliess1995flatness, calvet1988feedforward_quasi_linSys, deluca1998feedback} in the presence of input disturbances.

\subsection{Contributions}
The two main contributions of this work are highlighted as follows. First, we develop two methods for obtaining tight bounds on the peak output of a feedback-linearizable control-affine system in the presence of disturbances: a semi-definite programming (SDP) and a Lyapunov-equation approach. Second, we incorporate the output peak bounds in a governor control design, proving joint safety and stability for output trajectory tracking, and show an application of the complete approach on an Ackermann-drive robot. Our construction relies only on local obstacle information from onboard sensors and, hence, can be used in priori unknown environments.

\section{Notation}
\label{sec:notation}

Define $\intset{a,b}\coloneqq \crl{a, a+1, \ldots, b}$ for $a, b \in \mathbb{N}$, $a < b$ and let $\intset{b} \coloneqq \intset{1,b}$. Let $\CL_\Bf V(\Bx)$ denote the Lie derivative of a function $V(\Bx)$ along a vector field $\Bf(\Bx)$. We use $\col{\Bv_1, \Bv_2, \ldots, \Bv_n} \coloneqq [\Bv_1^\top, \Bv_2^\top, \ldots, \Bv_n^\top]^\top$ to denote vertical stacking of scalars, vectors, or matrices. Let $\mathbb{S}^n_{>0}$ denote the space of $n \times n$ positive definite matrices. The quadratic norm induced by $\BS \in \mathbb{S}^n_{>0}$ will be denoted by $\norm{\Bp}_\BS \coloneqq \sqrt{\Bp^\top \BS \Bp}$ and $\norm{\Bp}$ will be the Euclidean norm. Denote the distance from a point $\Bp$ to a set $\CA$ as $d_\BS(\Bp, \CA) \coloneqq \inf_{\Bq \in \CA} \norm{\Bq - \Bp}_\BS$ and let $d(\Bp, \CA)$ be the distance in the Euclidean norm.

\section{Problem Statement}
\label{sec:problem}

Consider a robot operating in an unknown environment $\C{W} \subseteq \mathbb{R}^{\nPath}$ with obstacle space denoted by $\CO \subset \C{W}$. Denote the free space by $\C{F} \coloneqq \C{W} \setminus \C{O}$ and its interior by $\mathring{\CF}$. The robot is modeled as a control-affine nonlinear dynamical system,
\begin{equation}
\label{eq:sys_pf}
\begin{aligned}
\dot{\Bx} &= \Bf(\Bx) + \BG(\Bx) (\Bu + \Bw), \qquad \Bx(t_0) = \Bx_0 \\
\By &= \Bh(\Bx)
\end{aligned}
\end{equation}
where $\Bx \in \calX\subset \bbR^{\nRbtState}$  denotes the state variables, $\Bu \in \bbR^{\nInput}$ represents the control input, $\Bw \in \bbR^{\nInput}$ is the input noise, $\By \in \bbR^{\nOutput}$ are the system outputs\footnote{To simplify the presentation, we consider a system with the same number of inputs and outputs. Our approach relies on feedback linearization, which can generally be applied when fewer or more outputs are available. When more outputs are available, it is sufficient that the rectangular decoupling matrix $\BM(\Bx)$ defined in Def.~\ref{def:vector_relative_degree} has full column rank and its left-inverse $\BM(\Bx)^\dagger = \prl{\BM(\Bx)^{\top} \BM(\Bx)}^{-1} \BM(\Bx)^{\top}$ can be used to define the feedback linearization. If fewer outputs than inputs are available, new ones can be defined following the procedure in~\cite[Prop.~9.16]{sastry}.}, and $\Bf(\Bx)$, $\BG(\Bx)$, $\Bh(\Bx)$ are smooth functions. We make the following assumptions.

\begin{assumption} \label{assump:bounded_noise}
The input noise is bounded, i.e., $\norm{\Bw} \leq \delta_\Bw$ for some $\delta_\Bw > 0$.
\end{assumption}

\begin{assumption}
\label{assump:feedback_linearizable}
System~\eqref{eq:sys_pf} satisfies the conditions for feedback linearization in~\cite[Lemma.~5.2.1]{isidori1995nonlinear} almost everywhere in $\calX$. 
\end{assumption}

Our approach relies on feedback linearization to transform the high-order nonlinear system dynamics~\eqref{eq:sys_pf} to a space defined by the system outputs $\By$, where the dynamics are simple. This is possible for many nonlinear systems, including differentially flat robots such as differential-drive, Ackermann-drive, fixed-wing aerial, and quadrotor robots \cite{Franch2009_ECC, murray1995differential, sreenath2013geometric}. For example, a simple geometric variable such as position may be used as the output $\bfy$ of an acceleration-controlled Ackermann-drive as shown in Sec.~\ref{sec:application}. Our objective is to design a control policy $\bfu(t)$ so that the system outputs $\bfy(t)$ follow a desired path in free space without violating safety constraints, i.e, $\By(t) \in \calF$ for all $t \geq t_0$.

\begin{definition}
\label{def:path}
A \emph{path} is a piecewise-continuous function $\Br: \brl{0,1} \mapsto \mathring{\CF}$ mapping a path-length parameter $\sigma \in \brl{0,1}$ to the interior of the free space. The start of a path is $\Br(0) \in \mathbb{R}^{\nPath}$; the end of a path is $\Br(1) \in \mathbb{R}^{\nPath}$.
\end{definition}

Many control techniques require a reference trajectory of the same dimension as the original system, which places significant burden on planning systems to ensure dynamic feasibility. Moreover, this decoupling into kinodynamic planning and stable tracking cannot guarantee safety because the trajectory time allocation provided by the planner cannot be altered by the controller. Our approach relies on low-dimensional geometric paths that may be computed efficiently by standard planning algorithms~\cite{RRTstar, lavalle2006planning,ARAstar}. The problem considered in this paper is stated below.


\begin{problem}
\label{Prob:Main_Prob}
Given a path $\Br$ such that $\Br(0) = \bfy(0)$, design a control policy $\bfu(t)$ so that the constrained state $\bfy(t)$ of~\eqref{eq:sys_pf} is asymptotically steered to $\Br(1)$, while remaining safe, i.e., $\bfy(t) \in \calF$ for all $t \geq 0$.
\end{problem}

We use feedback linearization under Assumption~\ref{assump:feedback_linearizable} to find a diffeomorphism that transforms the nonlinear system in~\eqref{eq:sys_pf} to a linear one. Then, we examine the impact of the input noise $\Bw(t)$ on the linearized system under Assumption~\ref{assump:bounded_noise} and derive accurate, yet efficiently computable, bounds on the output under stabilizing control. Finally, we introduce a first-order virtual \textit{governor} system that enables safe and adaptive tracking of the desired path. The real system aims to converge locally to the governor, while the governor adapts itself using the system trajectory bounds to progress along the path without endangering safety. The structure of the proposed safe adaptive controller is illustrated in Fig~\ref{fig:rgs_structure_rss}. Our approach achieves safe navigation with an Ackermann-drive robot in an unknown environment, relying only on onboard depth sensing.

\section{Feedback Linearization}
\label{sec:feedback_linearization}

Assumption~\ref{assump:feedback_linearizable} ensures system~\eqref{eq:sys_pf} has vector relative degree $\boldsymbol{\rho} \in \mathbb{R}^{\nInput}$ with $\sum_{i=1}^{\nInput} \rho_i = \nRbtState$.

\begin{definition}
\label{def:vector_relative_degree}
A control-affine nonlinear system~\eqref{eq:sys_pf} with outputs $y_i = h_i(\Bx)$ for $i \in \intset{\nInput}$ has \emph{vector relative degree} $(\rho_1, \ldots, \rho_{\nInput})$ in a region $\calD \subset \CX$ if, for all $\Bx \in \calD$:
	\begin{enumerate}
		\item $\CL_{\Bg_j} \CL_\Bf^{k-1} h_i(\Bx) = 0$ for all $j\in \intset{\nInput}$, $i \in \intset{\nInput}$, $k \in \intset{\rho_i}$
		\item the $\nInput \times \nInput$ \emph{decoupling matrix} below is invertible:
		\begin{equation}
		\label{eq:decoupling_matrix}
		\scaleMathLine[0.9]{\BM(\Bx) =\begin{bmatrix}
		\CL_{\Bg_1} \CL_\Bf^{\rho_1-1} h_1(\Bx) 		&\cdots 		&\CL_{\Bg_p} \CL_\Bf^{\rho_1-1} h_1(\Bx) \\
		\vdots											&\ddots			&\vdots \\
		\CL_{\Bg_1} \CL_\Bf^{\rho_\nOutput-1} h_q(\Bx) 		&\cdots 		&\CL_{\Bg_p} \CL_\Bf^{\rho_\nOutput-1} h_q (\Bx)
		\end{bmatrix}}
		\end{equation}
	\end{enumerate}
\end{definition}

Since the relative degree satisfies $\sum_{i=1}^{\nInput} \rho_i = \nRbtState$, we can define a diffeomorphic transformation of the state.

\begin{proposition}[{\cite[Prop. 5.1.2]{isidori1995nonlinear}}]
\label{prop:new_coordinates}
Assume the control-affine nonlinear system~\eqref{eq:sys_pf} with outputs $\By = \Bh(\Bx)\in \mathbb{R}^{\nInput}$ has vector relative degree $\boldsymbol{\rho} \in \mathbb{R}^{\nInput}$ in $\calD$ such that $\sum_{i=1}^{\nInput} \rho_i = \nRbtState$. Then, $\Phi(\Bx) \coloneqq \col{\Bxi^1, \ldots, \Bxi^{\nInput}}$, where:
\begin{equation}
\Bxi^i \coloneqq \begin{bmatrix}
\xi_1^i \\
\xi_2^i \\
\vdots \\
\xi_{\rho_i}^i
\end{bmatrix} = \begin{bmatrix}
\phi_1^i(\Bx) \\
\phi_2^i(\Bx) \\
\vdots \\
\phi_{\rho_i}^i(\Bx)
\end{bmatrix} = \begin{bmatrix}
h_i(\Bx) \\
\CL_\Bf h_i(\Bx) \\
\vdots \\
\CL_\Bf^{\rho_i-1} h_i(\Bx) 
\end{bmatrix}
\end{equation}
defines a local diffeomorphism.
\end{proposition}

Proposition~\ref{prop:new_coordinates} allows us to express~\eqref{eq:sys_pf} in new coordinates $\Bz = \Phi(\bfx)$, where the control input $\Bu$ can be chosen to cancel the nonlinearities and obtain a \emph{linear system} subject to state-dependent noise.
\begin{proposition} \label{prop:fdk_lin}
Applying change of coordinates, $\Bz = \Phi(\bfx)$, and control input:
\begin{equation}
\label{eq:fdk_ctrl_law_nf}
\Bu = \BM^{-1}(\Bx) \brl{-\Bn(\Bx) + \Bv},
\end{equation}
where $\BM(\Bx)$ is the decoupling matrix in~\eqref{eq:decoupling_matrix} and $\Bn(\Bx)\coloneqq \col{\CL_{\Bf}^{\rho_1}h_1,\ldots, \CL_{\Bf}^{\rho_\nInput}h_\nInput}$, to the system in~\eqref{eq:sys_pf} leads to linear dynamics with state-dependent noise:
\begin{align}
	\dot{\Bz} &= \BA \Bz + \BB \Bv + \BB_{\Bw}(\Phi^{-1}(\Bz)) \Bw, \quad \Bz(t_0) = \Bz_0 := \Phi(\Bx_0) \notag\\
	\By &= \BC \Bz \label{eq:sys_quasi_lin}
\end{align}
where $\BA \in \mathbb{R}^{\nRbtState \times \nRbtState}$, $\BB \in\mathbb{R}^{\nRbtState \times \nInput}$, $\BC \in \mathbb{R}^{\nInput\times \nRbtState}$ are block-diagonal matrices with blocks in Brunovsky Canonical Form and $\BB_\Bw(\Bx) = \col{\Bzero, \Bm_1(\Bx), \ldots, \Bzero, \Bm_\nInput(\Bx)} \in \mathbb{R}^{\nRbtState\times\nInput}$ where $\Bm_i(\Bx)$ is the $i$-th row of the decoupling matrix~\eqref{eq:decoupling_matrix}.
\end{proposition}
\begin{proof} \label{proof:fdk_lin}
See Appendix~\ref{app:proof_feedback_linearization}.
\end{proof}

\section{Safe and Stable Regulation}
\label{sec:regulation}

In this section, we choose a stabilizing controller $\Bv = -\BK \Bz$ for the linearized system~\eqref{eq:sys_quasi_lin} and analyze its peak output $\max_{t \geq t_0} \norm{\By(t)}^2$ in order to guarantee safety. The closed-loop system is
\begin{equation}
\label{eq:sys_quasi_lin_closed_loop}
\begin{aligned}
\dot{\Bz} &= \prl{\BA - \BB \BK} \Bz + \BB_{\Bw}(\Bx) \Bw\\
\By &= \BC \Bz
\end{aligned}
\end{equation}
where $\BK$ is chosen so that $(\BA - \BB \BK)$ is Hurwitz and, in the absence of disturbances, the origin is a globally exponentially stable (G.E.S.) equilibrium. To proceed, we also need a bound on the peak norm of the decoupling matrix.

\begin{assumption} \label{asp:Bw_bound}
There exists a finite upper bound, $\norm{\BM(\Bx)}_2 \leq \gamma(\Bx_0)$, on the norm of the decoupling matrix.
\end{assumption}

It will be convenient to reorder the states through a permutation matrix $\BT$ such that $\tRbt \coloneqq \BT \Bz = \col{\tilde{\Bz}_1, \tilde{\Bz}_2}$, with $\tRbt_1 = \By$. Due to Prop.~\ref{prop:new_coordinates} such a permutation $\BT$ always exists. Note that
\begin{equation}
\BT \BB_\Bw(\Bx)\Bw = \begin{bmatrix}
	\Bzero \\
	\BM(\Bx)
\end{bmatrix} \Bw = \begin{bmatrix}
	\Bzero \\
	\gamma(\Bx_0) \delta_\Bw \BI
\end{bmatrix}\frac{\BM(\Bx)\Bw}{\gamma(\Bx_0) \delta_\Bw}
\end{equation}
and system~\eqref{eq:sys_quasi_lin_closed_loop} with reordered states becomes:
\begin{equation} \label{eq:sys_relaxed}
\begin{aligned}
\dot{\tilde{\Bz}} &= \bar{\BA} \tilde{\Bz} + \bar{\BB} \bar{\Bw}, \qquad \tRbt(t_0) = \tRbt_0 =  \BT \Phi(\Bx_0) \\
\By &= \bar{\BC} \tilde{\Bz}
\end{aligned}
\end{equation}
where $\bar{\BA} \coloneqq \BT (\BA- \BB \BK)  \BT ^{-1}$, $\bar{\BB} \coloneqq \col{\Bzero,  \gamma(\Bx_0) \delta_\Bw \BI}$, $\bar{\Bw} \coloneqq \BM(\Bx)\Bw / \gamma(\Bx_0) \delta_\Bw$ and $\bar{\BC}  \coloneqq \BT\BC =  \brl{\BI, \Bzero}$. We are interested in finding the output peak of~\eqref{eq:sys_relaxed} along the system trajectory:
\begin{equation}
\label{eq:output_norm}
\eta(t_0) := \max_{t \geq t_0} \|\By(t)\|_{\BS}^2
\end{equation}
for some $\BS \in \mathbb{S}_{>0}^{\nOutput}$ to be specified later. Finding the exact output bound may be challenging~\cite{abedor1996linear}. Instead, we compute the tightest upper bound on $\eta(t_0)$ over an ellipsoid outer approximation of all possible system trajectories of~\eqref{eq:sys_relaxed} starting at $\tRbt_0$. Denote an ellipsoid in $\mathbb{R}^{\nRbtState}$, centered at $\Bp \in \mathbb{R}^{\nRbtState}$ and defined by $\BP \in \mathbb{S}_{>0}^{\nRbtState}$, as:
\begin{equation} \label{eq:ellipsoid_unit_energy_form}
\CE(\BP, \Bp)  \coloneqq \crl{\Bq \in \bbR^{\nRbtState}  \mid \prl{\Bq - \Bp}^{\top} \BP \prl{\Bq - \Bp} \leq 1}
\end{equation}

\begin{definition}\label{def:inv_ellipsoid}
An ellipsoid $\CE$ is \emph{positively invariant} for dynamical system~\eqref{eq:sys_relaxed} if $\tilde{\Bz}(t_0) \in \CE$ implies $\tilde{\Bz}(t) \in \CE$, for all $t \geq t_0$ and every system trajectory $\tilde{\Bz}(t)$.
\end{definition}

Thus, if $\CE_{inv}$ is a positively invariant ellipsoid for~\eqref{eq:sys_relaxed}, we can obtain an upper bound on $\eta(t_0)$ as follows:
\begin{equation}
\label{eq:peak_upperbound}
\eta(t_0) \leq \max_{\Bq \in \CE_{inv}} \|\bar{\BC}\Bq\|^2_{\BS}
\end{equation}
The following result allows us to find a positively invariant ellipsoid, centered at the equilibrium point of~\eqref{eq:sys_relaxed}, in the presence of bounded disturbances.

\begin{lemma}[{\cite{boyd_LMI_book}}] \label{lemma:invariant_ellipsoid}
Consider system \eqref{eq:sys_relaxed} with bounded disturbance $\norm{\bar{\Bw}}_2 \leq 1$. The ellipsoid $\CE(\BP,\Bzero)$ is positively invariant if and only if there exists a real scalar $\alpha \geq 0$ such that
	\begin{equation}
	\begin{bmatrix} \label{eq:LMI_PIE}
	\bar{\BA}^{\top} \BP + \BP \bar{\BA} + \alpha \BP 	& \BP \bar{\BB} \\
	\bar{\BB}^{\top} \BP	& -\alpha \BI
	\end{bmatrix} \preceq \Bzero
	\end{equation}
\end{lemma}

Conditions~\eqref{eq:peak_upperbound} and~\eqref{eq:LMI_PIE} can be combined to obtain a tight bound on the output peak $\eta(t_0)$.


\begin{theorem} \label{thm:output_bound_sdp}
Consider system~\eqref{eq:sys_relaxed} with bounded disturbance $\norm{\bar{\Bw}}_2 \leq 1$. The output of any trajectory, starting at $\tilde{\Bz}_0$, is bounded for all $t \geq t_0$ as follows:
\begin{equation}
\label{eq:output_bound_sdp}
\|\By(t)\|_{\BS}^2 \leq \eta(t_0) \leq  \delta_\eta(\alpha^*;\tilde{\Bz}_0) \coloneqq \min_{\alpha \in (0, \bar{\alpha})} \delta_\eta(\alpha;\tilde{\Bz}_0)
\end{equation}
where $\bar{\alpha} \coloneqq - 2 \max \prl{\text{real} \prl{\text{spec}(\bar{\BA})}}$ and $\delta_\eta (\alpha;\tilde{\Bz}_0)$ is the solution to the semi-definite program (SDP):
\begin{equation}
\label{eq:SDP_UB}
\begin{aligned}
& \underset{\BP, \delta}{\text{minimize}} & &  \delta \\
& \text{subject to}
& &  \begin{bmatrix}
\bar{\BA}^{\top} \BP + \BP \bar{\BA} + \alpha \BP 	& \BP \bar{\BB} \\
\bar{\BB}^{\top} \BP	& -\alpha \BI
\end{bmatrix} \preceq \Bzero  \\
&  & & \begin{bmatrix}
\BP 	& \bar{\BC}^{\top}\BS^{1/2} \\
\BS^{1/2}\bar{\BC}	& \delta \BI
\end{bmatrix} \succeq \Bzero, \;\;  \BP \succ \Bzero\\
& & & \tilde{\Bz}_0^{\top} \BP \tilde{\Bz}_0 \leq 1
\end{aligned}
\end{equation}
\scaleLine{Moreover, $\delta_\eta(\alpha^*;\Bzero)$ is the ultimate bound for $\|\By(t)\|_{\BS}^2$ as $t \to \infty$.}
\end{theorem}
\begin{proof}
See Appendix~\ref{app:proof_SDP_output_bound}.
\end{proof}

Thm.~\ref{thm:output_bound_sdp} shows that an output peak bound may be obtained via a two-level optimization. The lower-level is an $\alpha$-parameterized SDP. The upper-level is a scalar minimization of a scalar function over a bounded region $\alpha \in (0, \bar{\alpha})$, which may be approached with a global optimization method such as the Brent's algorithm~\cite[Ch.~5]{brent_algorithm}. While the bound computed by this optimization is very accurate, computation speed may be a concern for high-frequency control applications. An alternative looser bound on $\|\By(t)\|_{\BS}^2$ may be obtained by solving a series of algebraic Lyapunov equations.


\begin{theorem}\label{thm:output_bound_lyap}
Consider system~\eqref{eq:sys_relaxed} with bounded disturbance $\norm{\bar{\Bw}}_2 \leq 1$. The output of any trajectory, starting at $\tilde{\Bz}_0$, is bounded for all $t \geq t_0$ as follows:
\begin{equation}
\|\By(t)\|_{\BS}^2 \leq \eta(t_0) \leq  \hat{\delta}_\eta(\alpha^*;\tilde{\Bz}_0) \coloneqq \min_{\alpha \in (0, \bar{\alpha})} \hat{\delta}_\eta(\alpha;\tilde{\Bz}_0)
\end{equation}
where $\bar{\alpha} \coloneqq - 2 \max \prl{\text{real} \prl{\text{spec}(\bar{\BA})}}$ and
\begin{equation}
\label{eq:output_bound_lyap}
\hat{\delta}_\eta (\alpha;\tilde{\Bz}_0) \coloneqq \lambda_{\max} \prl{\BS^{1/2}\bar{\BC} \BQ_{\alpha} \bar{\BC}^{\top}\BS^{1/2}} \max \crl{\tilde{\Bz}_0^T \BQ_{\alpha}^{-1} \tilde{\Bz}_0, 1}
\end{equation}
where $\BQ_\alpha$ is the solution of Lyapunov equation:
\begin{equation}
\bar{\BA} \BQ + \BQ \bar{\BA}^{\top} + \alpha \BQ + \alpha^{-1} \bar{\BB} \bar{\BB}^{\top} = \Bzero \label{eq:inv_lyap}
\end{equation}
\scaleLine{Moreover, $\hat{\delta}_\eta(\alpha^*;\Bzero)$ is an ultimate bound for $\|\By(t)\|_{\BS}^2$ as $t \to \infty$.}
\end{theorem}
\begin{proof}
	See Appendix \ref{app:proof_lyap_output_bound}
\end{proof}

Thm.~\ref{thm:output_bound_lyap} does not jointly consider the initial condition constraint while minimizing the output bound. As a result, $\hat{\delta}_\eta(\alpha_{lyap}^*;\tilde{\Bz}_0 ) \geq \delta_\eta(\alpha_{sdp}^*;\tilde{\Bz}_0)$. The two $\alpha^*$ values are usually not the same.  when $\norm{\tilde{\Bz}_0}$ is not negligible compared to $\norm{\bar{\Bw}}$. Albeit looser, the bound in Thm.~\ref{thm:output_bound_lyap} is computationally much cheaper to obtain than solving a series of SDPs, as required by Thm.~\ref{thm:output_bound_sdp}.

\section{Safe and Stable Tracking using a Reference Governor}
\label{sec:tracking}

\begin{figure}[t]
	\centering
	\includegraphics[width=\linewidth]{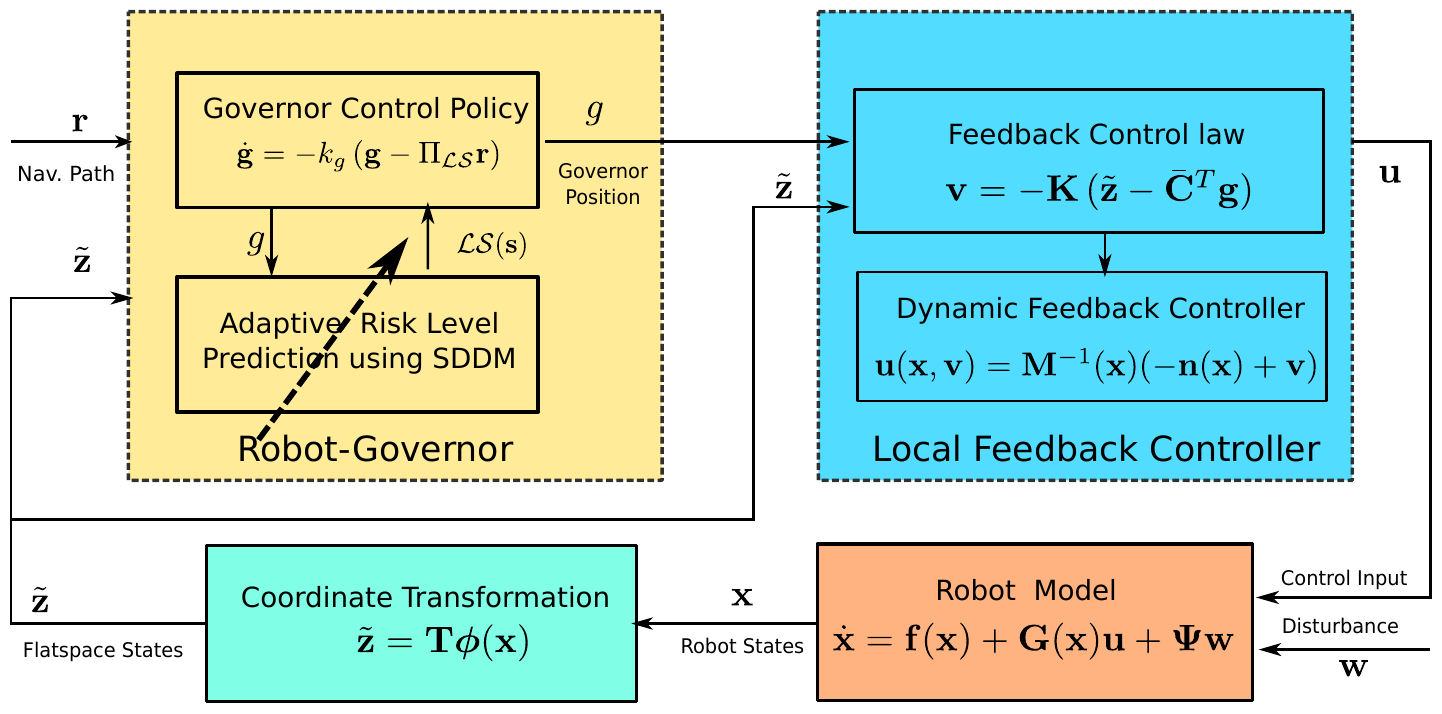}
	\caption{Structure of robot-governor system for safe path-following. Given a feasible path in working space, the reference governor evaluates safety level of the system with respect to local environment and updates governor position and sends it out to local feedback controller which drives the robot stable and safely towards the goal along navigation path.}
	\label{fig:rgs_structure_rss}
\end{figure}

Instead of regulating the linearized system~\eqref{eq:sys_relaxed} to a fixed equilibrium, we will use the output trajectory bound provided by Thm.~\ref{thm:output_bound_sdp} or Thm.~\ref{thm:output_bound_lyap} to decide whether it is safe to move the equilibrium point along the desired path $\Br$. We introduce a \emph{reference governor}~\cite{garone2016_ERG, kolmanovsky2014ref_cmd_gov}, a virtual first-order system $\dgp = \Bu_\Bg$ whose state $\Bg \in \mathbb{R}^{\nOutput}$ behaves as a real-time adaptive output reference signal, slowing down if the output bound may violate safety and speeding up otherwise. The goal is to design a governor control policy $\Bu_\Bg(t)$ that tracks the desired output path $\Br$, while not allowing the output bound $\delta_\eta(\alpha^*(t);\tilde{\Bz}(t))$ to exceed the distance to the obstacles $\CO$. Consider an augmented robot-governor system with state $\stateRgs \coloneqq \col{\tRbt, \gp}$ defined as:
\begin{equation} \label{eq:rgs_system}
\begin{aligned}
\begin{bmatrix}
\dtRbt\\
\dgp
\end{bmatrix} &= \begin{bmatrix}
\bar{\BA} \prl{\tRbt - \bar{\BC}^{\top} \gp}
 + \bar{\BB} \bar{\Bw} \\
\Bu_\Bg
\end{bmatrix} \\
\bar{\By} &= \bar{\BC} (\tRbt - \bar{\BC}^{\top} \gp) 
\end{aligned}
\end{equation}
where the equilibrium point of~\eqref{eq:sys_relaxed} has been shifted to $\col{\gp,\Bzero}$ so that the output of the linearized system tracks $\gp$. We define a local safe zone for the augmented system based on the output bound $\delta_\eta(\alpha^*;\tilde{\Bz})$ in~\eqref{eq:output_bound_sdp}.
\begin{definition}
	\label{def:LS}
	A \emph{local safe zone} is a time-varying set that at time $t$ depends on the robot-governor state $\stateRgs$ as follows:
	\begin{equation}
	\label{eq:LS}
	\LS(\stateRgs) \coloneqq \crl{ \Bq \in \CF \mid d_{\BS}^2(\Bq, \gp) \leq  \max \prl{0, \DeltaE(\stateRgs)} },
	\end{equation}
	where $\DeltaE(\stateRgs) \coloneqq d_\BS^2 \prl{\gp, \CO} - \delta_\eta(\alpha^*;\tilde{\Bz}-\bar{\BC}^{\top}\gp) - \epsilon_E$ is a measure of leeway to safety violation, $\delta_\eta$ is a bound on the output peak of $\bar{\By}$ measured by a quadratic norm defined by $\BS \in \bbS^{\nPath}_{>0}$, and $\epsilon_E > 0$ is a small constant to accommodate numerical errors.
\end{definition}

The first term, $d_\BS^2 \prl{\gp, \CO}$, in the definition of $\DeltaE(\stateRgs)$ estimates the (quadratic) distance from the desired system equilibrium $\gp$ to the nearest obstacles. The second term, $\delta_\eta(\alpha^*;\tilde{\Bz}-\bar{\BC}^{\top}\gp)$ evaluates the maximum deviation of the system output $\bar{\By}$ from the equilibrium point $\gp$. The value $\Delta E(\stateRgs)$ reflects the level of risk of the robot-governor system at the current state $\stateRgs$. Intuitively, it is a measure of remaining energy that the system or the governor may use before safety is endangered. The requirement that $\Delta E(\Bs) \geq 0$ only places a constraint on the magnitude of $\norm{\dgp}$, so $\dgp / \norm{\dgp}$ is a degree of freedom that can be utilized to make $\gp$ asymptotically tend toward a desired goal region. We use any available energy to move the governor $\gp$ along the desired path $\Br$. 

\begin{definition}
	\label{def:LPG}
	A \emph{local projected goal} is the furthest point along the path $\Br$ that intersects with the local safe zone $\LS(\stateRgs)$:
	\begin{equation}
	\label{eq:local_projected_goal}
	\lpg(\stateRgs) = \Pi_{\LS(\stateRgs)} \Br \coloneqq \max_{\sigma \in [0,1]} \crl{ \Br(\sigma) \mid  \Br(\sigma) \in \LS(\stateRgs)}.
	\end{equation}
\end{definition}

We use the informal notation $\Pi_{\LS(\stateRgs)} \Br$ because determining the local projected goal is equivalent to finding the metric projection of $\Br$ on the boundary of the local safe zone. We define the \emph{governor control policy} $\Bu_\Bg$ to track the local projected goal:
\begin{equation} \label{eq:gov_ctrl} 
\Bu_\Bg(t) = - k_g (\gp(t) - \lpg(t)) 
\end{equation}
where $k_g > 0$ is a control gain for the governor controller.

Note that the local safe zone set is an ellipsoid whose size is determined by the free energy $\DeltaE(\stateRgs)$ and whose shape is determined by the positive definite matrix $\BS$. If the volume of $\LS(\stateRgs)$ is large (no nearby obstacles), the local projected goal will be far along $\Br$, leading to fast governor and, in turn, system motion. The shape of $\LS(\stateRgs)$ can be chosen via $\BS$ to obtain fast tracking even if there are nearby obstacles that do not interfere with the direction of motion. More precisely, $\BS$ may be chosen adaptively to be elongated in the direction from the system output $\By$ to the equilibrium point $\gp$. 

Finally, we derive conditions that guarantee the robot remains safe and ultimately bounded around a static equilibrium and that the control policy~\eqref{eq:gov_ctrl} applied to the robot-governor system~\eqref{eq:rgs_system} tracks the desired path $\Br$ safely, converging to a small region around the goal, determined by the input noise bound $\delta_{\Bw}$.

\begin{theorem} \label{thm:safety_static_gov}
Let $\col{\tRbt_0, \gp_0}$ be any initial state for~\eqref{eq:rgs_system} with $\bar{\By}, \gp_0 \in \CF$ and let $\Bu_\gp(t) \equiv \Bzero$ so that the governor remains static, i.e., $\gp(t) \equiv \gp_0$. Suppose that the following safety condition is satisfied:
\begin{equation}
\label{eq:safety_condition}
\delta_\eta(\alpha^*;\tRbt_0-\bar{\BC}^\top\gp_0) \leq d^2_{\BS}(\gp_0, \CO),
\end{equation}
where $\delta_\eta$ is an upper bound on $\norm{\bar{\By}(t)}_{\BS}^2$, e.g., obtained from Thm.~\ref{thm:output_bound_sdp} or Thm.~\ref{thm:output_bound_lyap} for any $\BS \in \bbS^{\nOutput}_{>0}$. Then, output trajectory is collision free, i.e., $\By(t) \in \CF$ for all $t \geq t_0$, and ultimately bounded by $\delta_\eta(\alpha^*;\Bzero)$.
\end{theorem}

\begin{proof}
	See Appendix \ref{app:proof_static_gov} 
\end{proof}

\begin{definition}
The set of \emph{safe states} of the robot-governor system~\eqref{eq:rgs_system} includes all states with strictly positive free energy and associated local safe zone with non-empty intersection with the path $\Br$:
	\begin{equation}
	\calS \coloneqq \crl{\Bs \mid \DeltaE(\Bs) > 0, \exists\, \sigma \in \brl{0, 1} \text{ s.t. } \Br(\sigma) \in \LS(\Bs)}.
	\end{equation}
\end{definition}

\begin{definition}
A \emph{goal region} of size $\epsilon$ is a set of states $\Bs= \col{\tRbt,\gp}$ of the robot-governor system~\eqref{eq:rgs_system} such that the output $\By = \bar{\BC}\tRbt$ of the original nonlinear system~\eqref{eq:sys_pf} is at a distance at most $\epsilon$ from the path end point $\Br(1)$:
	\begin{equation}
	\calG(\epsilon) \coloneqq \crl{\Bs = \col{\tRbt,\gp} \mid d_{\BS}^2(\By, \Br(1)) \leq \epsilon,\; \By = \bar{\BC}\tRbt}.
	\end{equation}
\end{definition}

\begin{theorem} \label{thm:moving_gov_safe}
Suppose that the path $\Br$ satisfies the following safety condition:
\begin{equation}
\min_{\sigma \in \brl{0,1}} d_{\BS}(\Br(\sigma), \CO) > \sqrt{\lambda_{\min}(\BS)(\delta_\eta(\alpha^*;\Bzero) + \epsilon_E)}
\end{equation}
for some constant $\epsilon_E > 0$ and bound $\delta_\eta(\alpha^*;\Bzero)$ on the peak output $\max_{t \geq t_0} \|\By(t)\|_{\BS}^2$ of system~\eqref{eq:sys_pf} under the feedback-linearizing controller in~\eqref{eq:sys_quasi_lin_closed_loop}. Then, the closed-loop robot-governor system~\eqref{eq:rgs_system} with governor control policy $\Bu_\Bg$ in~\eqref{eq:gov_ctrl} is asymptotically steered from any safe initial state $\col{\BT\Phi(\Bx_0),\gp_0} \in \calS$ to a goal region $\calG(\delta_\eta(\alpha^*;\Bzero)/\lambda_{\min}(\BS))$ and the robot output trajectory is collision-free for all time, i.e., $\By(t) \in \CF$ for all $t\geq t_0$.
\end{theorem}
\begin{proof}
See appendix~\ref{app:proof_rgs}.
\end{proof}

\section{Application: Ackermann-Drive Robot}
\label{sec:application}

We demonstrate the proposed reference-governor controller on a ground wheeled robot with an Ackermann steering mechanism. The state of the robot consists of its position $(x,y)$, orientation $\carorientation$ and steering angle $\carsteering$, while its input is the rear-axle driving speed $\carspeed$ and the steering angular velocity $\omega$. The kinematic model~\cite{deluca1998feedback} describing the robot's motion is:
\begin{equation} \label{eq:sys_ack_kinematic_mdl}
\begin{bmatrix}
\dot{x} \\
\dot{y} \\
\dot{\carorientation} \\
\dcarsteering
\end{bmatrix} = \begin{bmatrix}
\cos \carorientation \\
\sin \carorientation \\
\frac{\tan \carsteering}{\carlengthtt}  \\ 
0 
\end{bmatrix} \carspeed + \begin{bmatrix}
0 \\
0 \\
0 \\
1
\end{bmatrix} \omega
\end{equation}
where $\carlengthtt$ is the distance between the front and rear axles. The model assumes that there is no wheel slip and the two front wheels have the same steering angle. In the literature, this model is also referred to as a bicycle model~\cite{deluca1998feedback}, single-track model~\cite{ackermann1993robust}, or Ackermann vehicle model~\cite{Franch2009_ECC}.
 
The system~\eqref{eq:sys_ack_kinematic_mdl} is not static feedback linearizable but it is linearizable via dynamic extension~\cite{fliess1995flatness, isidori1995nonlinear}. Speed $\carspeed$ and acceleration $a \coloneqq \dot{v}$ are added to the system state $\stateAck = \col{x, y, \carorientation, \carsteering, \carspeed, a}$, while jerk $j \coloneqq \dot{a}$ is added as a control input $\controlAck = \col{j, \omega_{\carsteering}}$. The augmented system can be written in the control-affine form~\eqref{eq:sys_pf} with bounded input disturbance $\Bw  \coloneqq \col{w_{1}, w_{2}}$:
\begin{equation}\label{eq:aug_ack_dynamics}
\begin{aligned} 
\begin{bmatrix}
\dot{x} \\
\dot{y} \\
\dcarorientation \\
\dcarsteering \\
\dcarspeed \\
\dot{a}
\end{bmatrix} &= \begin{bmatrix}
\carspeed \cos \carorientation \\
\carspeed \sin \carorientation \\
\frac{\carspeed}{\carlengthtt} \tan \carsteering \\
0 \\
a \\
0 
\end{bmatrix} + \begin{bmatrix}
0 & 0 \\
0 & 0 \\
0 & 0 \\
0 & 1 \\
0 & 0 \\
1 & 0 
\end{bmatrix} \prl{\begin{bmatrix} 
j \\
\omega
\end{bmatrix} + \begin{bmatrix} 
w_1 \\
w_2
\end{bmatrix}} \\
\end{aligned}
\end{equation}

Choosing position as an output, $\outputAck = \Bh(\stateAck) = \col{x, y}$, makes the system output-feedback linearizable:
\begin{equation}
\label{eq:M_ack_n_ack}
\scaleMathLine[0.89]{\dddoutputAck = \underbrace{\begin{bmatrix}
\cos \carorientation & -\frac{\carspeed^2 \sin \carorientation}{\carlengthtt \cos^2 \carsteering} \\
\sin \carorientation & \frac{\carspeed^2 \cos \carorientation}{\carlengthtt \cos^2 \carsteering} \\	
\end{bmatrix}}_{\BM(\stateAck)}\controlAck + \underbrace{\begin{bmatrix}
- 3 a \dot{\carorientation} \sin \carorientation  
- \carspeed \dot{\carorientation}^2 \cos \carorientation \\
3 a \dot{\carorientation} \cos \carorientation  
- \carspeed \dot{\carorientation}^2 \sin \carorientation
\end{bmatrix}}_{\Bn(\stateAck)}}
\end{equation}
Note that $\det \BM(\stateAck) = \frac{\carspeed^2}{\carlengthtt \cos^2 \carsteering}$, so $\BM(\stateAck)$ is non-singular if $\carspeed \neq 0$ and $\carsteering \neq \pm \frac{\pi}{2}$. Due to the geometry of the steering mechanism, the wheels can never be perpendicular to the vehicle body, while the non-zero speed requirement can be satisfied by initializing our controller just a small initial speed. Hence, Assumption~\ref{assump:feedback_linearizable} is satisfied. The system has vector relative degree $\bfrho = (3,3)$ and $\sum_i \rho_i = 6$, which equals the number of states. By Prop.~\ref{prop:new_coordinates}, the system is full-state input-output linearizable with new coordinates $\Bz = \Phi(\Bx) \coloneqq \col{x, \dot{x}, \ddot{x}, y, \dot{y}, \ddot{y}}$. By Prop.~\ref{prop:fdk_lin}, applying the control input $\controlAck = \BM^{-1}(\stateAck)\brl{\Bv - \Bn(\stateAck)}$ leads to the linear system with state-dependent noise in~\eqref{eq:sys_quasi_lin}.

Instead of assuming a conservative bound on the input noise $\Bw(t)$, we propose theres exist an operating profile that relates the allowable speed $\carspeed$ and steering angle $\carsteering$ of the car. The profile specifies that as the speed $\carspeed$ is increasing, the allowable steering angle $\carsteering$ is decreasing. Since $\carsteering \neq \pm \pi/2$, the operating envelope ensures that some upper bound $\beta < \infty$ exists on $\carspeed^2/(l \cos^2 \carsteering)$. From \eqref{eq:M_ack_n_ack}, we can compute:
\begin{equation}
\scaleMathLine[0.89]{\norm{\BM(\Bx)}_2 = \lambda_{\max}^{1/2}( \textbf{diag} (1, \carspeed^4/(l \cos^2 \carsteering)^2) \leq \max \prl{1, \beta},}
\end{equation}
showing that Assumption~\ref{asp:Bw_bound} is satisfied with $\gamma(\Bx_0)\equiv \max \prl{1, \beta}$. Following the approach in Sec.~\ref{sec:regulation}, we choose a stabilizing control gain $\BK$ for the linear system in~\eqref{eq:sys_quasi_lin} and re-order the states as $\tRbt = \col{x, y, \dot{x}, \dot{y}, \ddot{x}, \ddot{y}}$, arriving at the system in~\eqref{eq:sys_relaxed}. The reminder of the control design, including the local safe zone and projected goal computations in~\eqref{eq:LS} and~\eqref{eq:local_projected_goal}, is agnostic to the original system dynamics.

\section{Evaluation}
\label{sec:evaluation}
This section evaluates the output peak bounds provided by Thm.~\ref{thm:output_bound_sdp} and Thm.~\ref{thm:output_bound_lyap} and demonstrates the complete control design for safe ackermann-drive robot navigation in a simulated environment.


\subsection{Output Prediction for Stochastic Linear Systems}
\label{sec:output_bounds}
Consider the linearized system~\eqref{eq:sys_relaxed} corresponding to the Ackermann-drive robot described in Sec.~\ref{sec:application}. Fig.~\ref{fig:PIE_example} compares the invariant ellipsoids $\crl{\bfy \mid \|\bfy\|_{\BS}^2 \leq \delta}$ containing the system output, predicted by Thm.~\ref{thm:output_bound_sdp} and Thm.~\ref{thm:output_bound_lyap} for two different choices of $\BS$. We can observe that all four ellipsoids are valid outer approximations of the space of output trajectories. The SDP-based method (Thm.~\ref{thm:output_bound_sdp}) leads to tighter bounds than the Lyapunov-equation-based method (Thm.~\ref{thm:output_bound_lyap}) regardless of how the quadratic output norm is defined. Using a quadratic norm aligned with the initial direction of motion of the system, $-\bfy(0)/\|\bfy(0)\|$, leads to a large improvement on the tightness of the SDP-based output peak bound. The Lyapunov-equation-based method does not benefit significantly from a specific choice of $\BS$ because the bound estimation does not take the initial condition into account when minimizing the size of the invariant ellipsoid.


\begin{figure}[t]
	\centering
	\includegraphics[width=0.49\linewidth,trim=8mm 8mm 0mm 0mm, clip]{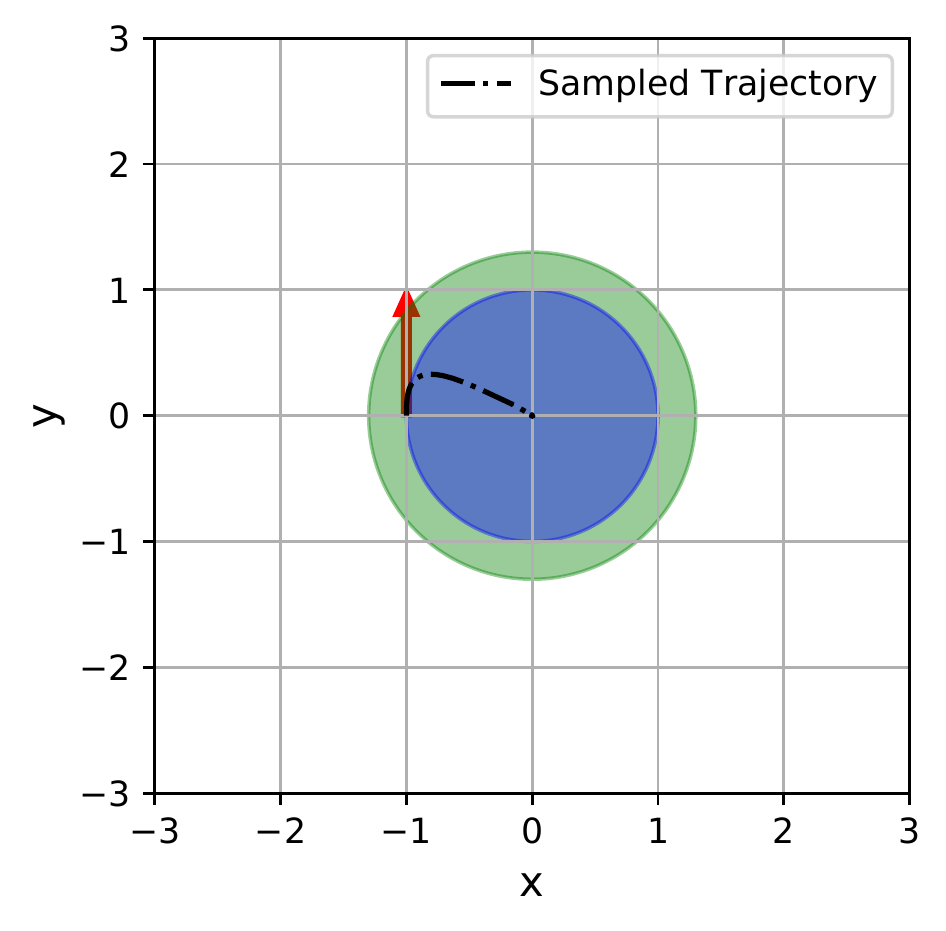}%
	\hfill %
	\includegraphics[width=0.49\linewidth,trim=8mm 8mm 0mm 0mm, clip]{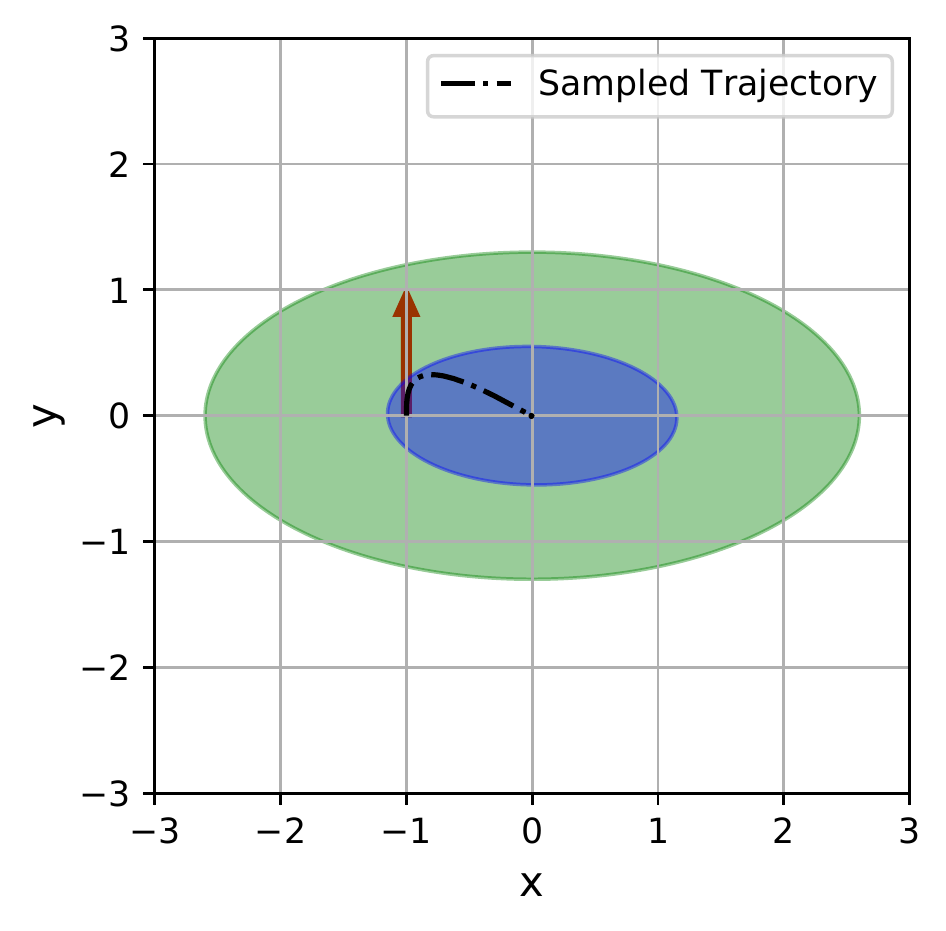}%
\caption{Trajectory output bounds comparison between SDP- (blue, Thm.~\ref{thm:output_bound_sdp}) and Lyapunov-equation-based (green, Thm.~\ref{thm:output_bound_lyap}) methods using two different distance metrics: Euclidean norm (left) and quadratic norm with $\BS =  \textbf{diag}(1,4)$ (right). The closed-loop linear system~\eqref{eq:sys_relaxed} has poles at $[-1, -1,-3,-3,-5,-5]$, with initial condition $\tRbt_0 = \col{-1, 0, 0, 1, 0, 0}$ (red arrow). The disturbance norm and speed-steering ratio bounds were $\delta_\Bw = 0.1$ and $\beta = 10$.}
	\label{fig:PIE_example}
\end{figure}

\subsection{Simulated Ackermann-Drive Navigation}
\label{sec:simulation}
Next, we demonstrate the performance of the complete control design on the Ackermann-drive robot in a simulated navigation tasks. The robot is operating in an unknown environment and relies on a simulated Lidar scanner to measure the Euclidean distance $d^2(\gp(t), \calO)$ from the governor $\gp(t)$ to the obstacles $\calO$ (see Fig.~\ref{fig:corridor_sim_snapshots}). The path $\Br$ is re-planned using the $A^{*}$ algorithm~\cite{ARAstar} from the current governor position $\Bg(t)$ to a fixed goal location within an occupancy grid map~\cite[Ch.~9]{ProbabilisticRoboticsBook} constructed from the lidar scans over time. It is important to note that the controller does not rely on the map of the environment. It only tracks the geometric path provided by $A^{*}$, relying on the lidar scanner to evaluate the local safe zone $\LS(\stateRgs)$ in~\eqref{eq:LS} and determine the the local projected goal $\lpg(\stateRgs)$ in~\eqref{eq:local_projected_goal}. The governor regulates the tracking speed using the output peak bounds proposed in Sec.~\ref{sec:regulation}, slowing down when new obstacles appear to ensure that the real robot system remains safe and stable. Fig.~\ref{fig:safety_risk_quantization.pdf} shows that the output trajectory bounds remain valid over time and that a safe distance from the obstacles is maintained. We can see that the bounds computed by the SDP-based and Lyapunov-equation-based methods are quite similar and both are valid overapproximations of the space of paths from the robot $\By(t)$ position ot the governor position $\gp(t)$. As noted earlier, the Lyapunov-equation-based method is much more computationally efficient, while the SDP method offers very high accuracy, especially with an appropriate choice of a (time-varying) quadratic norm scaling $\BS$. Since $d^2(\gp, \calO) \geq d^2(\gp, \By)$ throughout the simulation, the triangle inequality guarantees that $d(\By, \calO) \geq 0$ and hence the robot is safe, $\By(t) \in \calF$, for all time $t \geq t_0$.

\begin{figure}[t]
	\centering
	\includegraphics[width=1\linewidth]{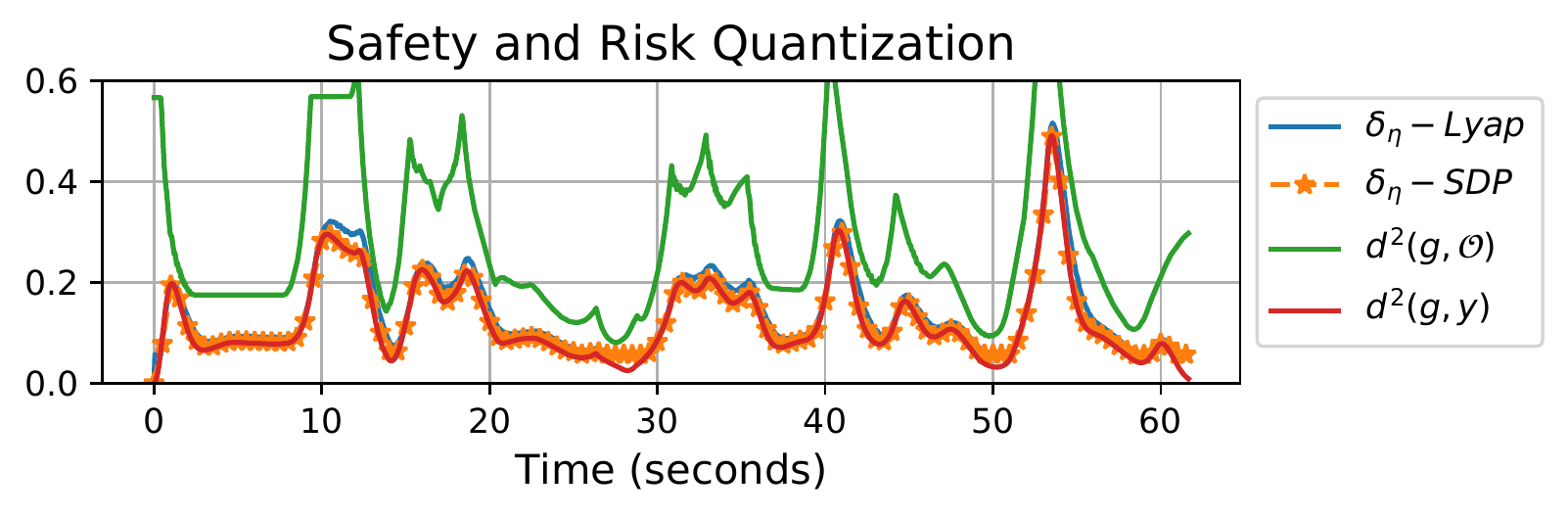}
	\caption{Illustration of the output bounds $\delta_\eta$ (orange) and $\hat{\delta}_\eta$ (blue) predicted over time by Thm.~\ref{thm:output_bound_sdp} and Thm.~\ref{thm:output_bound_lyap}, respectively, for the simulation in Fig.~\ref{fig:corridor_sim_snapshots}. An Ackermann-drive robot is controlled using the governor controller with local safe zone computed based on $\hat{\delta}_\eta$. The SDP-based bound is computed for comparison. The plot shows that the upper bounds are always valid as they remain above the distance $d^2\prl{\gp,\By}$ (red) between the robot position $\By$ and the governor $\gp$ and that both ensure safety as they remain below the distance $d^2\prl{\gp, \calO}$ (green) from the governor $\gp$ to the obstacles $\calO$.
	$\delta_\Bw = 2$ and $\beta = 10$ with closed-loop poles for $\tRbt$-subsystem at $[-5.784,-5.784, -0.858+1.1569j, -0.858-1.1569j, -0.858+1.1569j, -0.858-1.1569j]$ }
	\label{fig:safety_risk_quantization.pdf}
\end{figure}

\begin{figure}[t]
	\centering
	\includegraphics[width=1.0\linewidth]{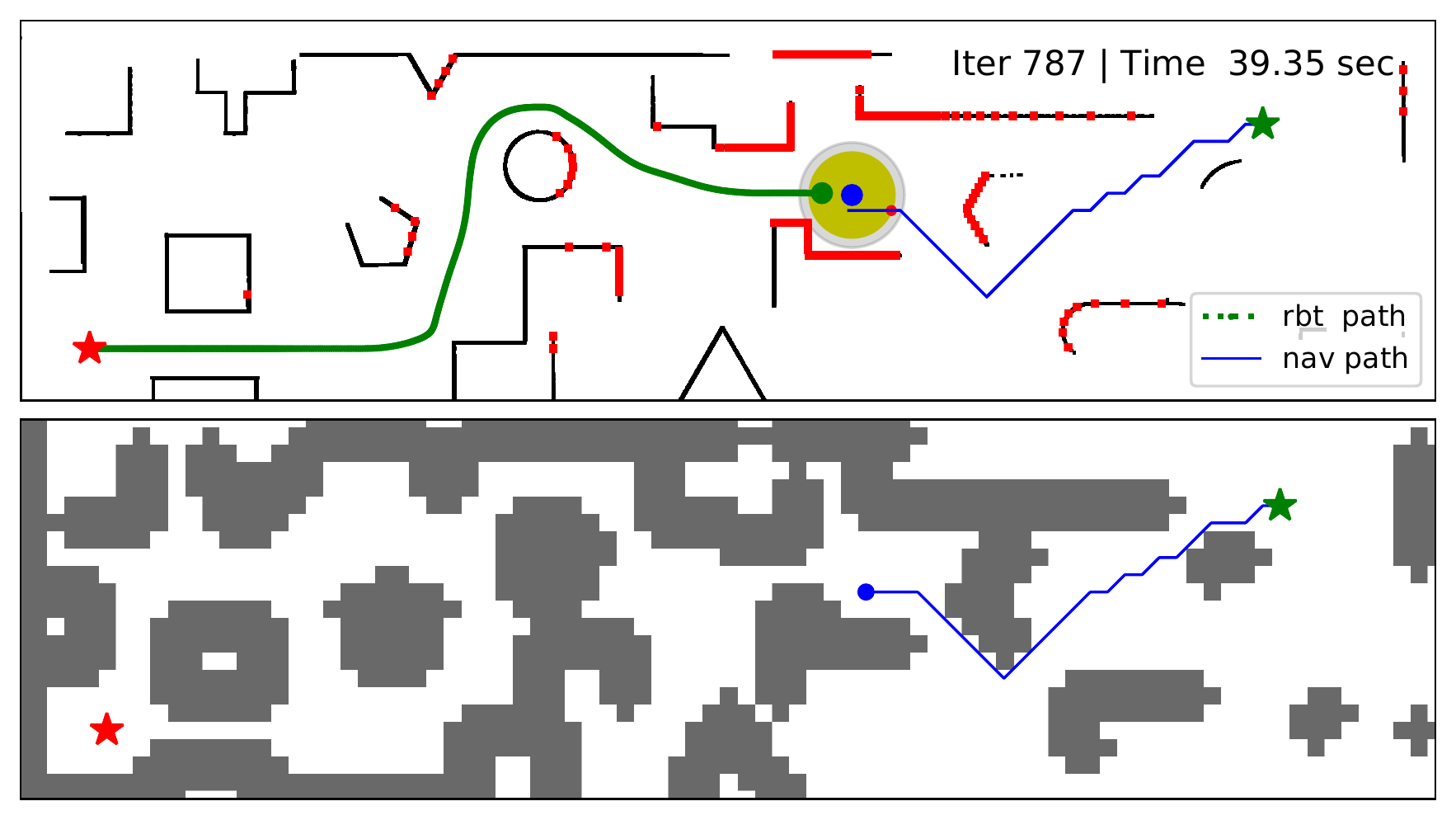}%
	\caption{Snapshots of the robot-governor system navigating a simulated environment. Streaming lidar scan measurements (red dots) are used to update an occupancy grid map (black lines and white regions in top plot) of the unknown environment. An Ackermann-drive robot (green dot) follows a virtual governor (blue dot) whose motion is modulated based on the local energy zone (yellow ball) and the distance to obstacles (gray ball). A navigation path (blue line) is periodically replanned using an $A^*$ planner and an inflated occupancy map (bottom plot).}
	\label{fig:corridor_sim_snapshots}
\end{figure}




\section{Conclusion}
\label{sec:conclusion}
This paper presented a safe, stable, and fast path-following control design for non-linear systems subject to bounded disturbances. The controller relies on feedback-linearization and requires only an output reference signal (rather than a high-order reference model), allowing the use of an efficient geometric path planner for autonomous navigation with differentially flat robots, such as cars or quadrotors. The design guarantees joint stability and safety using only local obstacle information, making it suitable for navigation in unknown environments. Future work will focus on additional hardware experiments pushing the performance limits of the controller in challenging environments and will incorporate ideas from control-barrier-function-based designs which avoid the limitations of feedback linearization. The simplicity of the proposed safety criteria also creates a promising avenue for research on safe online learning of robot dynamics.

\appendices

\section{} \label{app:proof_feedback_linearization} 
Applying the change of variables $\Bz = \Phi(\Bx)$ transforms the system~\eqref{eq:sys_pf} to \emph{normal form}:
\begin{equation} \label{eq:FL_normal_form_nsy}
\begin{aligned}
\dot{\xi}_1^i &= \xi_2^i + \sum_{j=1}^{m} (u_j+w_j) \CL_{\Bg_j} h_i(\Bx) &\\
&\vdots  &\\
\dot{\xi}_{\rho_i}^i  &=  \CL_\Bf^{\rho_i} h_i + \sum_{j=1}^{m}  (u_j+w_j) \CL_{\Bg_j}\CL_\Bf^{\rho_i - 1} h_i(\Bx) &  \\
y_i &= \xi_1^i  \qquad \qquad \qquad \qquad \text{for all}\quad i \in \intset{m} 
\end{aligned}
\end{equation}
From Def.~\ref{def:vector_relative_degree}, $\CL_{\Bg_j} \CL_\Bf^{k-1} h_i(\Bx) = 0$ for all $j\in \intset{\nInput}$, $i \in \intset{\nInput}$, $k \in \intset{\rho_i}$, leading to
\begin{equation}
\begin{aligned}
\dot{\xi}_k^i &= \xi_{k+1}^i, \quad \text{for all } i \in \intset{\nInput}, k \in \intset{\rho_i-1},\\
\begin{bmatrix}
\dot{\xi}_{\rho_1}^1 \\
\vdots\\
\dot{\xi}_{\rho_\nOutput}^{\nOutput}
\end{bmatrix} &= \Bn(\Bx) + \BM(\Bx) \prl{\Bu  + \Bw},
\end{aligned}
\end{equation}
where $\Bn(\Bx) \coloneqq \col{\CL_{\Bf}^{\rho_1}h_1,\ldots, \CL_{\Bf}^{\rho_\nOutput}h_\nOutput}$ and $\BM(\Bx)$ is the decoupling matrix in  \eqref{eq:decoupling_matrix}. Applying $\Bu =  \BM^{-1}(\Bx) \brl{-\Bn(\Bx) + \Bv}$, leads to the following system in coordinates $\Bz = \Phi(\Bx)$:
\begin{equation}
\begin{aligned}
\dot{\Bz} &= \BA \Bz + \BB \Bv + \BB_\Bw(\Bx) \Bw \\
\By &= \BC \Bz,
\end{aligned}
\end{equation}
where $\BB_\Bw(\Bx) = \col{\Bzero, \Bm_1(\Bx), \ldots, \Bzero, \Bm_\nInput(\Bx)} \in \mathbb{R}^{\nRbtState\times\nInput}$ where $\Bm_i(\Bx)$ is the $i$-th row of the decoupling matrix~\eqref{eq:decoupling_matrix}. The matrices $\BA$, $\BB$, $\BC$ are block-diagonal, $\BA = \textbf{diag}(\BA_{1}, \ldots, \BA_{\nOutput})$, $\BB = \textbf{diag}(\Bb_{1}, \ldots, \Bb_{\nOutput})$, $\BC = \textbf{diag}(\Bc_{1}, \ldots, \Bc_{\nOutput})$ with elements $\BA_{i} \in \bbR^{\rho_i \times \rho_i}$, $\Bb_{i} \in \bbR^{\rho_i}$, $\Bc_{i} \in \bbR^{\rho_i}$ in Brunovsky Canonical Form (BCF):
\begin{equation}
\BA_{i} = \begin{bmatrix}
0 & 1 & 0 &\ldots &0 \\
0 & 0 & 1 &\ldots &0 \\
\vdots &\vdots &\vdots &\cdots &1 \\
0 &0 &0 &\ldots &0
\end{bmatrix} \;\;\; 
\Bb_{i} = \begin{bmatrix}
0 \\
\vdots \\
0 \\
1
\end{bmatrix} \;\;\; \Bc_{i} =
\begin{bmatrix}
1 \\
0 \\
\vdots \\
0
\end{bmatrix}^{\top}.
\end{equation}


\section{} \label{app:proof_SDP_output_bound}
From lemma~\ref{lemma:invariant_ellipsoid}, find a invariant ellipsoid is equivalent to an find a feasible solution for constraint \eqref{eq:LMI_PIE} for some $\alpha \geq 0$. Let $\eta(t_0, \tilde{\Bz}_0)$ represents the upper bound on $\norm{\By}_\BS^2$ over all possible realization of trajectories subject to $\bar{\Bw}(t)$ starting with initial state at $\tilde{\Bz}_0$. Without loss of generality, consider $\tilde{\Bz}_0 = \Bzero$. Let $\CE(\BP_\alpha)$ be an $\alpha$-parameterized invariant ellipsoid. Without loss of generality, let us assume $\BS = \BI$ first, i.e., $\By = \bar{\BC} \tilde{\Bz}$, we have:
%
\begin{equation*}
\begin{aligned}
\eta(t_0, \Bzero) &\leq \sup_{\crl{\Bq \mid \Bq^T \BP_\alpha \Bq^T \leq 1}} \Bq^T \bar{\BC}^T \bar{\BC} \Bq & \\
&= \sup_{\crl{\Bv \mid \norm{\Bv} \leq 1}} \Bv^T \BP_\alpha^{-\frac{1}{2}} \bar{\BC}^T \bar{\BC} \BP_\alpha^{-\frac{1}{2}} \Bv & \text{change of variable}\\
&= \lambda_{\max}(\BP_\alpha^{-\frac{1}{2}} \bar{\BC}^T \bar{\BC} \BP_\alpha^{-\frac{1}{2}}) & \text{Rayleigh Quotient}
\end{aligned}
\end{equation*}
Obtaining the smallest upper bound on $\eta$ is equivalent to solve the following optimization problem
\begin{align} \label{eq:opt_trick}
& \underset{\BP_\alpha, \delta_\eta}{\text{minimize}} & &  \delta_\eta \\
& \text{subject to}
& &   \lambda_{\max}(\BP_\alpha^{-\frac{1}{2}} \bar{\BC}^T \bar{\BC} \BP_\alpha^{-\frac{1}{2}}) \leq \delta_\eta \label{eq:lamda_max_bound}
\end{align}
Using \emph{Schur complement} and \emph{Sylvester's law of inertia}, inequality \eqref{eq:lamda_max_bound} can be expressed  as a LMI:
\begin{equation}
\begin{bmatrix}
\BP_\alpha 	& \bar{\BC}^T \\
\bar{\BC}	& \delta_\eta \BI
\end{bmatrix} \succeq \Bzero
\end{equation}  
Considering initial condition $\tilde{\Bz} \neq 0$, we need to incorporate additional constraint $\tilde{\Bz}_0^T \BP_\alpha \tilde{\Bz}_0 \leq 1$. In order to find the invariant set with smallest upper bound $\delta_\eta$, we need to solve the optimization problem over all $\alpha \in (0, \bar{\alpha})$. The upper bound of  $\alpha$ is obtained by pre-solving $\bar{\BA}^T \BP + \BP \bar{\BA} + \alpha \BP \preceq \Bzero$
of LMI \eqref{eq:LMI_PIE}, which requires $(\bar{\BA} + \frac{\alpha}{2} \BI)$ to be Hurwitz.
The ultimate bound of $\norm{\By(t)}^2$ can be obtained by solving this optimization problem without initial value constraint , since the linear system is G.E.S. the effect of initial condition will becomes negligible as $t \rightarrow \infty$.  Since $\BS \in \bbS^{\nOutput}_{>0}$, replacing the $\bar{\BC}$ with $\BS^{1/2} \bar{\BC}$ leads to bounds on $\norm{\By}_{\BS}^2$ straightforwardly. \qed

\section{} \label{app:proof_lyap_output_bound}
This result follows from the work of Abedor et al.~\cite{abedor1996linear} and Brockman and Corless~\cite{brockman1998quadratic}. If $\CE(\BP,\Bzero)$ is positively invariant for~\eqref{eq:sys_relaxed} with $\tilde{\Bz}_0 = \Bzero$, then Lemma~\ref{lemma:invariant_ellipsoid} implies that there exists $\alpha \in (0, \bar{\alpha})$ such that $\crl{\Bq \in \mathbb{R}^{\nRbtState} \mid \Bq^{\top}\BQ_\alpha^{-1}\Bq\leq 1} \subseteq \crl{\Bq \in \mathbb{R}^{\nRbtState} \mid \Bq^{\top}\BP\Bq\leq 1}$. Hence, the smallest invariant ellipsoids are generated by solutions to~\eqref{eq:inv_lyap} as $\alpha$ sweeps $(0, \bar{\alpha})$. Still assuming $\tilde{\Bz}_0 = \Bzero$, from~\eqref{eq:peak_upperbound}, the output peak satisfies
\begin{equation}
\eta(t_0) \leq \!\!\!\max_{\Bq \in \CE(\BQ_\alpha^{-1},\Bzero)} \!\!\|\bar{\BC}\Bq\|^2_{\BS} = \lambda_{\max}\!\prl{\BS^{1/2}\bar{\BC} \BQ_{\alpha} \bar{\BC}^{\top}\BS^{1/2}}
\end{equation}
Sec.~5 in~\cite{brockman1998quadratic} shows how to incorporate a non-zero initial condition in the bound to obtain the result in~\eqref{eq:output_bound_lyap}.\qed

\section{} \label{app:proof_static_gov}
If the governor is static at $\gp_0$ and $\Bw \equiv \Bzero$, the robot-governor system~\eqref{eq:rgs_system} is G.E.S. with respect to the equilibrium $\lpg = \col{\Bg_0, \Bzero_{\nRbtState - \nOutput}, \Bg_0}$ since $\bar{\BA}$ in the linear time-invariant system: $\dtRbt = \bar{\BA} (\tRbt - \bar{\BC}^T \gp_0)$ is Hurwitz. From  Thm.~\ref{thm:output_bound_sdp} or Thm.~\ref{thm:output_bound_lyap}, when the $\tRbt$ subsystem is subject to bounded disturbance $\norm {\bar{\Bw}} \leq 1$, the output norm $\norm{\bar{\By}(t)}_{\BS}^2$ is bounded by $\delta_\eta(\alpha^{*}; \tRbt_0)$ for all $t \geq t_0$. From the ellipsoid definition in~\eqref{eq:ellipsoid_unit_energy_form}, $\norm{\bar{\By}(t)}_{\BS}^2 \leq \delta_\eta(\alpha^{*}; \tRbt_0)$ is equivalent to $\bar{\By}(t) \in \CE \prl{\delta_\eta^{-1}(\alpha^{*}; \tRbt_0) \BS, \Bzero}$, which implies $\By(t) \in \CE \prl{\delta_\eta^{-1}(\alpha^{*}; \tRbt_0)\BS, \gp_0}$. In turn, the safety condition in~\eqref{eq:safety_condition} implies $\CE \prl{\delta_\eta^{-1}(\alpha^{*}; \tRbt_0)\BS, \gp_0} \subseteq  \CE \prl{d^{-2}_{\BS}(\gp_0, \CO)\BS, \gp_0}$. Finally, by definition of $d^2_{\BS}(\gp_0, \CO)$, we can assure $\By(t) \in \CF$. As $t \rightarrow \infty$, the size of the ellipsoid $\CE \prl{\delta_\eta^{-1}(\alpha^{*}; \tRbt_0)\BS, \gp_0}$ will be decreasing and since the noise-free equilibrium point is G.E.S., $\delta_\eta(\alpha^*;\Bzero)$ will be the ultimate bound for $\norm{\bar{\By}}^2$.\qed



\section{} \label{app:proof_rgs}
From Thm.~\ref{thm:safety_static_gov}, we know that if the governor is static at $\gp_0$, the robot-governor system~\eqref{eq:rgs_system} will follow a collision-free trajectory and will be ultimately bounded around the equilibrium point $\lpg = \col{\Bg_0, \Bzero_{\nRbtState - \nOutput}, \Bg_0}$. The governor control policy in~\eqref{eq:gov_ctrl} allows the governor to move only when the interior of $\LS(\Bs)$ is nonempty. From Def.~\ref{def:LS}, this happens only if the safety condition is strictly satisfied, i.e., $\DeltaE(\Bs(t)) = d_{\BS}^2(\gp(t), \mathcal{O}) -  \delta_\eta(\alpha^*;\tRbt(t) - \bar{\BC}^\top \gp(t)) - \epsilon_E > 0$. Since $\By  \rightarrow \bar{\gp}$, $\delta_\eta (\alpha^*;\tRbt - \bar{\BC}^\top \gp) \rightarrow \delta_\eta (\alpha^*; \Bzero)$. By Def.~\ref{def:LPG}, $\Bg_0$ is always on the path, i.e., $\Bg_0 = \Br(\sigma)$ for some $\sigma \in \brl{0, 1}$, hence
$d_{\BS}^2(\Bg_0, \mathcal{O}) \geq d^2_\BS (\Br, \CO) \geq \lambda_{\min} (\BS) d^2(\Br, \CO)$
where $d_\BS(\Br, \CO) := \min_{\sigma \in [0,1]} d_\BS(\Br(\sigma),\CO)$. By assumption, the path is strictly feasible with clearance bound $d(\Br, \CO) > \sqrt{\lambda_{\min} (\BS) (\delta_\eta(\alpha^*;\Bzero) + \epsilon_E)}$, so $d_{\BS}^2(\Bg_0, \mathcal{O}) \geq \lambda_{\min} (\BS) (\delta_\eta(\alpha^*;\Bzero) + \epsilon_E)$ and $\DeltaE(\Bs)$ eventually becomes strictly positive. Once $\DeltaE > 0$, the set $\LS(\Bs)$ becomes an ellipsoid in free space with non-empty interior. Since initially $\Br(\sigma) \in \LS(\Bs_0)$ for some $\sigma \in [0,1]$, the local projected goal in~\eqref{eq:local_projected_goal} will be well defined and when $\LS(\Bs)$ grows, the projected goal will move further along the path $\Br$, i.e., the path length parameter $\sigma$ will increase. Assuming that $\epsilon_E > 0$ is sufficiently large to overcome numerical errors, $(\DeltaE(\Bs) + \epsilon_E)$ cannot suddenly become negative without crossing zero.  If $\DeltaE(\Bs) \downarrow 0$, the local energy zone $\LS(\Bs)$ shrinks to a point, i.e., $\LS(\Bs) = \crl{\gp}$, and hence the governor stops moving and waits until the robot catches up. When the governor is static, and since the safety condition in~\eqref{eq:safety_condition} is satisfied, Thm.~\ref{thm:safety_static_gov} again guarantees that the robot can approach the governor without collisions, increasing $\DeltaE(\Bs)$ in the process. Once $\DeltaE(\Bs)$ goes above $\Bzero$, the governor starts moving towards the goal again by chasing the projected goal. Note that the local projected goal always lies on the navigation path inside the free space and $d(\Br, \CO)$ has enough clearance. Hence, the robot-governor system cannot remain stuck at any configuration except in the ultimate bound region of its state $\Bs(t)$. Using LaSalle's Invariance Principle~\cite{khalil2002nonlinear} and Thm.~\ref{thm:safety_static_gov}, one can conclude that the output invariant set is $\CE \prl{\delta_\eta^{-1}(\alpha^*;\Bzero)\BS, \Br(1)}$ and therefore, $\By(t) \in \CE( \lambda_{\min} (\BS) \delta_\eta^{-1}(\alpha^*;\Bzero) \BI, \Br(1))$, i.e., $d(\By(t), \Br(1)) \leq \delta_\eta(\alpha^*;\Bzero)/\lambda_{\min} (\BS)$. \qed


%

%


\bibliographystyle{cls/abbrvplainnat}
\bibliography{bib/ref.bib}

\begin{thebibliography}{43}
\providecommand{\natexlab}[1]{#1}
\providecommand{\url}[1]{\texttt{#1}}
\expandafter\ifx\csname urlstyle\endcsname\relax
  \providecommand{\doi}[1]{doi: #1}\else
  \providecommand{\doi}{doi: \begingroup \urlstyle{rm}\Url}\fi

\bibitem[Abedor et~al.(1996)Abedor, Nagpal, and Poolla]{abedor1996linear}
J.~Abedor, K.~Nagpal, and K.~Poolla.
\newblock {A linear matrix inequality approach to peak-to-peak gain
  minimization}.
\newblock \emph{International Journal of Robust and Nonlinear Control},
  6\penalty0 (9-10):\penalty0 899--927, 1996.

\bibitem[Ackermann and Sienel(1993)]{ackermann1993robust}
J.~Ackermann and W.~Sienel.
\newblock {Robust yaw damping of cars with front and rear wheel steering}.
\newblock \emph{IEEE Transactions on Control Systems Technology}, 1\penalty0
  (1):\penalty0 15--20, 1993.

\bibitem[Aine et~al.(2016)Aine, Swaminathan, Narayanan, Hwang, and
  Likhachev]{MHA}
S.~Aine, S.~Swaminathan, V.~Narayanan, V.~Hwang, and M.~Likhachev.
\newblock {Multi-Heuristic A*}.
\newblock \emph{The International Journal of Robotics Research (IJRR)},
  35\penalty0 (1-3):\penalty0 224--243, 2016.

\bibitem[{Ames} et~al.(2019){Ames}, {Coogan}, {Egerstedt}, {Notomista},
  {Sreenath}, and {Tabuada}]{ADA-SC-ME-GM-KS-PT:19}
A.~D. {Ames}, S.~{Coogan}, M.~{Egerstedt}, G.~{Notomista}, K.~{Sreenath}, and
  P.~{Tabuada}.
\newblock Control barrier functions: Theory and applications.
\newblock In \emph{2019 18th European Control Conference (ECC)}, 2019.

\bibitem[Ames et~al.(2014)Ames, Galloway, Sreenath, and
  Grizzle]{CBF_ames2014rapidly}
A.~D. Ames, K.~Galloway, K.~Sreenath, and J.~W. Grizzle.
\newblock {Rapidly exponentially stabilizing control lyapunov functions and
  hybrid zero dynamics}.
\newblock \emph{IEEE Transactions on Automatic Control}, 59\penalty0
  (4):\penalty0 876--891, 2014.

\bibitem[Ames et~al.(2016)Ames, Xu, Grizzle, and Tabuada]{CBF_ames2016control}
A.~D. Ames, X.~Xu, J.~W. Grizzle, and P.~Tabuada.
\newblock {Control barrier function based quadratic programs for safety
  critical systems}.
\newblock \emph{IEEE Transactions on Automatic Control}, 62\penalty0
  (8):\penalty0 3861--3876, 2016.

\bibitem[Arslan and Koditschek(2017)]{Gov_ICRA17}
O.~Arslan and D.~E. Koditschek.
\newblock {Smooth extensions of feedback motion planners via reference
  governors}.
\newblock In \emph{IEEE International Conference on Robotics and Automation
  (ICRA)}, 2017.

\bibitem[Boyd et~al.(1994)Boyd, El~Ghaoui, Feron, and
  Balakrishnan]{boyd_LMI_book}
S.~Boyd, L.~El~Ghaoui, E.~Feron, and V.~Balakrishnan.
\newblock \emph{{Linear matrix inequalities in system and control theory}},
  volume~15.
\newblock Siam, 1994.

\bibitem[Brent(1973)]{brent_algorithm}
R.~P. Brent.
\newblock \emph{Algorithms for minimization without derivatives}.
\newblock Prentice-Hall, 1973.

\bibitem[Brockman and Corless(1998)]{brockman1998quadratic}
M.~L. Brockman and M.~Corless.
\newblock Quadratic boundedness of nominally linear systems.
\newblock \emph{International Journal of Control}, 71\penalty0 (6):\penalty0
  1105--1117, 1998.

\bibitem[Burridge et~al.(1999)Burridge, Rizzi, and
  Koditschek]{burridge1999sequential}
R.~Burridge, A.~Rizzi, and D.~Koditschek.
\newblock {Sequential composition of dynamically dexterous robot behaviors}.
\newblock \emph{The International Journal of Robotics Research (IJRR)},
  18\penalty0 (6):\penalty0 534--555, 1999.

\bibitem[Calvet and Arkun(1988)]{calvet1988feedforward_quasi_linSys}
J.-P. Calvet and Y.~Arkun.
\newblock {Feedforward and feedback linearization of non-linear systems with
  disturbances}.
\newblock \emph{International Journal of Control}, 48\penalty0 (4):\penalty0
  1551--1559, 1988.

\bibitem[De~Luca et~al.(1998)De~Luca, Oriolo, and Samson]{deluca1998feedback}
A.~De~Luca, G.~Oriolo, and C.~Samson.
\newblock {Feedback control of a nonholonomic car-like robot}.
\newblock In \emph{Robot motion planning and control}, pages 171--253.
  Springer, 1998.

\bibitem[Fliess et~al.(1995)Fliess, L{\'e}vine, Martin, and
  Rouchon]{fliess1995flatness}
M.~Fliess, J.~L{\'e}vine, P.~Martin, and P.~Rouchon.
\newblock {Flatness and defect of non-linear systems: introductory theory and
  examples}.
\newblock \emph{International journal of control}, 61\penalty0 (6):\penalty0
  1327--1361, 1995.

\bibitem[Franch and Rodriguez-Fortun(2009)]{Franch2009_ECC}
J.~Franch and J.~Rodriguez-Fortun.
\newblock {Control and trajectory generation of an Ackerman vehicle by dynamic
  linearization}.
\newblock In \emph{IEEE European Control Conference (ECC)}, pages 4937--4942,
  2009.

\bibitem[Gammell et~al.(2020)Gammell, Barfoot, and Srinivasa]{BITstar}
J.~D. Gammell, T.~D. Barfoot, and S.~S. Srinivasa.
\newblock {Batch Informed Trees (BIT*): Informed asymptotically optimal anytime
  search}.
\newblock \emph{The International Journal of Robotics Research (IJRR)}, 2020.

\bibitem[Gao et~al.(2018)Gao, Wu, Lin, and Shen]{SFC_FM}
F.~Gao, W.~Wu, Y.~Lin, and S.~Shen.
\newblock {Online safe trajectory generation for quadrotors using fast marching
  method and bernstein basis polynomial}.
\newblock In \emph{IEEE International Conference on Robotics and Automation
  (ICRA)}, pages 344--351, 2018.

\bibitem[Garone and Nicotra(2016)]{garone2016_ERG}
E.~Garone and M.~M. Nicotra.
\newblock {Explicit reference governor for constrained nonlinear systems}.
\newblock \emph{IEEE Transactions on Automatic Control (TAC)}, 61\penalty0
  (5):\penalty0 1379--1384, 2016.

\bibitem[Gawron and Michalek(2017)]{gawron2017vfo}
T.~Gawron and M.~M. Michalek.
\newblock {VFO feedback control using positively-invariant funnels for mobile
  robots travelling in polygonal worlds with bounded curvature of motion}.
\newblock In \emph{IEEE International Conference on Advanced Intelligent
  Mechatronics (AIM)}, pages 124--129, 2017.

\bibitem[Gawron and Michalek(2018)]{gawron_2018IROS_VFO}
T.~Gawron and M.~M. Michalek.
\newblock {Algorithmization of constrained motion for car-like robots using the
  VFO control strategy with parallelized planning of admissible funnels}.
\newblock In \emph{IEEE/RSJ International Conference on Intelligent Robots and
  Systems (IROS)}, pages 6945--6951, 2018.

\bibitem[Herbert et~al.(2017)Herbert, Chen, Han, Bansal, Fisac, and
  Tomlin]{herbert2017fastrack}
S.~L. Herbert, M.~Chen, S.~Han, S.~Bansal, J.~F. Fisac, and C.~J. Tomlin.
\newblock Fastrack: A modular framework for fast and guaranteed safe motion
  planning.
\newblock In \emph{IEEE Conference on Decision and Control (CDC)}, pages
  1517--1522, 2017.

\bibitem[Isidori(1995)]{isidori1995nonlinear}
A.~Isidori.
\newblock \emph{{Nonlinear control systems}}.
\newblock Springer, 1995.

\bibitem[Karaman and Frazzoli(2011)]{RRTstar}
S.~Karaman and E.~Frazzoli.
\newblock {Sampling-based algorithms for optimal motion planning}.
\newblock \emph{The International Journal of Robotics Research (IJRR)},
  30\penalty0 (7):\penalty0 846--894, 2011.

\bibitem[Khalil(2002)]{khalil2002nonlinear}
H.~Khalil.
\newblock \emph{Nonlinear systems}.
\newblock Prentice Hall, 2002.

\bibitem[Kolmanovsky et~al.(2014)Kolmanovsky, Garone, and
  Di~Cairano]{kolmanovsky2014ref_cmd_gov}
I.~Kolmanovsky, E.~Garone, and S.~Di~Cairano.
\newblock {Reference and command governors: A tutorial on their theory and
  automotive applications}.
\newblock In \emph{IEEE American Control Conference (ACC)}, 2014.

\bibitem[LaValle(1998)]{RRT}
S.~LaValle.
\newblock {Rapidly-exploring random trees: A new tool for path planning}.
\newblock Tr 98-11, Comp. Sci. Dept., Iowa State University, 1998.

\bibitem[LaValle(2006)]{lavalle2006planning}
S.~LaValle.
\newblock \emph{Planning Algorithms}.
\newblock Cambridge University Press, 2006.

\bibitem[Li et~al.(2016)Li, Littlefield, and Bekris]{SST}
Y.~Li, Z.~Littlefield, and K.~E. Bekris.
\newblock {Asymptotically Optimal Sampling-based Kinodynamic Planning}.
\newblock \emph{The International Journal of Robotics Research (IJRR)},
  35\penalty0 (5):\penalty0 528--564, 2016.

\bibitem[Likhachev et~al.(2004)Likhachev, Gordon, and Thrun]{ARAstar}
M.~Likhachev, G.~J. Gordon, and S.~Thrun.
\newblock {ARA*: Anytime A* with provable bounds on sub-optimality}.
\newblock In \emph{Advances in neural information processing systems (NIPS)},
  pages 767--774, 2004.

\bibitem[Liu et~al.(2017)Liu, Watterson, Mohta, Sun, Bhattacharya, Taylor, and
  Kumar]{SFC}
S.~Liu, M.~Watterson, K.~Mohta, K.~Sun, S.~Bhattacharya, C.~J. Taylor, and
  V.~Kumar.
\newblock {Planning dynamically feasible trajectories for quadrotors using safe
  flight corridors in 3-d complex environments}.
\newblock \emph{IEEE Robotics and Automation Letters (RA-L)}, 2\penalty0 (3),
  2017.

\bibitem[Majumdar and Tedrake(2017)]{Funnel_lib}
A.~Majumdar and R.~Tedrake.
\newblock {Funnel libraries for real-time robust feedback motion planning}.
\newblock \emph{The International Journal of Robotics Research (IJRR)},
  36\penalty0 (8):\penalty0 947--982, 2017.

\bibitem[Mason and Salisbury~Jr(1985)]{funnel_idea}
M.~T. Mason and J.~K. Salisbury~Jr.
\newblock {Robot hands and the mechanics of manipulation}.
\newblock 1985.

\bibitem[Murray et~al.(1995)Murray, Rathinam, and
  Sluis]{murray1995differential}
R.~M. Murray, M.~Rathinam, and W.~Sluis.
\newblock Differential flatness of mechanical control systems: A catalog of
  prototype systems.
\newblock In \emph{ASME international mechanical engineering congress and
  exposition}. Citeseer, 1995.

\bibitem[Rosolia and Borrelli(2019)]{rosolia2019sample}
U.~Rosolia and F.~Borrelli.
\newblock Sample-based learning model predictive control for linear uncertain
  systems.
\newblock \emph{arXiv preprint arXiv:1904.06432}, 2019.

\bibitem[Sastry(1999)]{sastry}
S.~Sastry.
\newblock \emph{{Nonlinear Systems: Analysis, Stability, and Control}}.
\newblock Springer, 1999.

\bibitem[Singh et~al.(2017)Singh, Majumdar, Slotine, and
  Pavone]{singh2017robust}
S.~Singh, A.~Majumdar, J.-J. Slotine, and M.~Pavone.
\newblock Robust online motion planning via contraction theory and convex
  optimization.
\newblock In \emph{IEEE International Conference on Robotics and Automation
  (ICRA)}, pages 5883--5890, 2017.

\bibitem[Sreenath et~al.(2013)Sreenath, Lee, and Kumar]{sreenath2013geometric}
K.~Sreenath, T.~Lee, and V.~Kumar.
\newblock Geometric control and differential flatness of a quadrotor uav with a
  cable-suspended load.
\newblock In \emph{IEEE Conference on Decision and Control}, pages 2269--2274.
  IEEE, 2013.

\bibitem[Tan et~al.(2008)Tan, Packard, et~al.]{sos_stability}
W.~Tan, A.~Packard, et~al.
\newblock {Stability region analysis using polynomial and composite polynomial
  Lyapunov functions and sum-of-squares programming}.
\newblock \emph{IEEE Transactions on Automatic Control (TAC)}, 53\penalty0
  (2):\penalty0 565, 2008.

\bibitem[Tedrake et~al.(2009)Tedrake, Manchester, Tobenkin, and
  Roberts]{lqr_tree_tedrake2009}
R.~Tedrake, I.~R. Manchester, M.~Tobenkin, and J.~W. Roberts.
\newblock {LQR-trees: Feedback Motion Planning via Sums-of-Squares
  Verification}.
\newblock \emph{The International Journal of Robotics Research (IJRR)}, 2009.

\bibitem[Thrun et~al.(2005)Thrun, Burgard, and Fox]{ProbabilisticRoboticsBook}
S.~Thrun, W.~Burgard, and D.~Fox.
\newblock \emph{Probabilistic Robotics}.
\newblock MIT Press Cambridge, 2005.

\bibitem[Webb and van~den Berg(2013)]{kinodynamic_rrts_webb2013}
D.~J. Webb and J.~van~den Berg.
\newblock {Kinodynamic RRT*: Asymptotically optimal motion planning for robots
  with linear dynamics}.
\newblock In \emph{IEEE International Conference on Robotics and Automation
  (ICRA)}, 2013.

\bibitem[Wu and Sreenath(2015)]{CBF_wu2015safety}
G.~Wu and K.~Sreenath.
\newblock {Safety-critical and constrained geometric control synthesis using
  control lyapunov and control barrier functions for systems evolving on
  manifolds}.
\newblock In \emph{American Control Conference (ACC)}, pages 2038--2044, 2015.

\bibitem[Wu and Sreenath(2016)]{CBF_quadrotor}
G.~Wu and K.~Sreenath.
\newblock {Safety-critical control of a planar quadrotor}.
\newblock In \emph{American Control Conference (ACC)}, pages 2252--2258, 2016.

\end{thebibliography}

\end{document}